\documentclass{llncs}

\usepackage{epsfig}
\usepackage{times}
\usepackage{gensymb}
\usepackage[english]{babel}
\usepackage{graphicx}
\usepackage{amssymb}
\usepackage{multirow}
\usepackage{subfigure}
\usepackage{xspace}

\newcommand{\morph}[1]{\langle #1 \rangle}
\newcommand{\mmorph}[1]{\ensuremath{\morph{#1}}}

\title{Morphing Planar Graph Drawings Optimally}

\date{}

\author{Patrizio Angelini$^1$, Giordano Da Lozzo$^1$, Giuseppe Di Battista$^1$, Fabrizio Frati$^2$, Maurizio Patrignani$^1$, Vincenzo Roselli$^1$
\institute{
$1$ Dipartimento di Ingegneria, Roma Tre University, Italy\\
\email{\{angelini,dalozzo,gdb,patrigna,roselli\}@dia.uniroma3.it}\\
$2$ School of Information Technologies, The University of Sydney, Australia\\
\email{fabrizio.frati@sydney.edu.au}
}}

\newcommand{\remove}[1]{}

\renewcommand{\int}{int}
\newcommand{\slope}{slope}

\renewenvironment{proof}
{{\bf Proof:}}{\hspace*{\fill}$\Box$\par\vspace{2mm}}

\newenvironment{proofx}
{{\bf Proof of Lemma~\ref{le:unidirectional}:}}{\hspace*{\fill}$\Box$\par\vspace{2mm}}

\begin{document}

\maketitle

\begin{abstract}
We\remove{show} provide an algorithm\remove{that computes} for computing a planar morph between any two planar straight-line drawings of any $n$-vertex plane graph in $O(n)$ morphing steps, thus improving upon the previously best known $O(n^2)$ upper bound. Further, we prove that our algorithm is optimal, that is, we show that there exist two planar straight-line drawings $\Gamma_s$ and $\Gamma_t$ of an $n$-vertex plane graph $G$ such that any planar morph between $\Gamma_s$ and $\Gamma_t$ requires $\Omega(n)$ morphing steps.
\end{abstract}

\section{Introduction} \label{se:introduction}

A {\em morph} is a continuous transformation between two topologically equivalent geometric objects. The study of morphs is relevant for several areas of computer science, including computer graphics, animation, and modeling. Many of the geometric shapes that are of interest in these contexts can be effectively described by two-dimensional planar graph drawings. Hence, designing algorithms and establishing bounds for morphing planar graph drawings is an important research challenge. We refer the reader to~\cite{ekp-ifmpg-03,fe-gdm-02,gs-gifpm-01,sg-cmcpt-01,sg-imct-03} for extensive descriptions of the applications of graph drawing morphs.

It has long been known that there always exists a {\em planar morph} (that is, a morph that preserves the planar topology of the graph at any time instant) transforming any planar straight-line drawing $\Gamma_s$ of a plane graph $G$ into any other planar straight-line drawing $\Gamma_t$\remove{of the same plane graph $G$} of $G$. However, the first proof of such a result, published by Cairns in 1944~\cite{c-dprc-44}, was ``existential'', meaning that no guarantee was provided on the complexity of the trajectories followed by the vertices during the morph. Almost 40 years later, Thomassen proved in~\cite{t-dpg-83} that a morph between $\Gamma_s$ and $\Gamma_t$ always exists in which vertices follow trajectories of exponential complexity (in the number of vertices of $G$). In other words, adopting a setting defined by Gr\"unbaum and Shepard~\cite{gs-tgopg-81} which is also the one we consider in this paper, Thomassen proved that there exists a sequence $\Gamma_s=\Gamma_1,\Gamma_2,\dots,\Gamma_k=\Gamma_t$ of planar straight-line drawings of $G$ such that, for every $1\leq i\leq k-1$, the {\em linear morph} transforming $\Gamma_i$ into $\Gamma_{i+1}$ is planar, where a linear morph moves each vertex at constant speed along a straight-line trajectory.

A breakthrough was recently obtained by Alamdari {\em et al.}\remove{at SODA~`13~\cite{aac-mpgdpns-13}. The authors proved} by proving that a planar morph between any two planar straight-line drawings of the same $n$-vertex connected plane graph exists in which each vertex follows a trajectory of polynomial complexity~\cite{aac-mpgdpns-13}. That is, Alamdari {\em et al.} showed an algorithm to perform the morph in $O(n^4)$ {\em morphing steps}, where a morphing step is\remove{ an intermediate} a linear morph. The $O(n^4)$ bound was shortly afterwards improved to $O(n^2)$ by Angelini {\em et al.}~\cite{afpr-mpgde-13}.

In this paper, we\remove{show} provide an algorithm to compute a planar morph with $O(n)$ morphing steps between any two planar straight-line drawings $\Gamma_s$ and $\Gamma_t$ of any $n$-vertex connected plane graph $G$. Further, we prove that our algorithm is optimal. That is, for every $n$, there exist two drawings $\Gamma_s$ and $\Gamma_t$ of the same $n$-vertex plane graph (in fact a path) such that any planar morph between $\Gamma_s$ and $\Gamma_t$ consists of $\Omega(n)$ morphing steps. To the best of our knowledge, no super-constant lower bound was previously known.

The schema of our algorithm is the same as in~\cite{aac-mpgdpns-13,afpr-mpgde-13}. Namely, we morph $\Gamma_s$ and $\Gamma_t$ into two drawings $\Gamma^x_s$ and $\Gamma^x_t$ in which a certain vertex $v$ can be contracted onto a neighbor $x$. Such contractions generate two straight-line planar drawings $\Gamma'_s$ and $\Gamma'_t$ of a smaller plane graph $G'$. A morph between $\Gamma'_s$ and $\Gamma'_t$ is recursively computed and suitably modified to produce a morph between $\Gamma_s$ and $\Gamma_t$. The main ingredient for our new bound is a drastically improved algorithm to morph $\Gamma_s$ and $\Gamma_t$ into $\Gamma^x_s$ and $\Gamma^x_t$. In fact, while the task of making $v$ contractible onto $x$ is accomplished with $O(n)$ morphing steps in~\cite{aac-mpgdpns-13,afpr-mpgde-13}, we devise and use properties of monotone drawings, level planar drawings, and hierarchical graphs to perform it with $O(1)$ morphing steps.

The idea behind the lower bound is that linear morphs can poorly simulate rotations, that is, a morphing step rotates an edge of an angle whose size is $O(1)$.\remove{ In fact, in a single linear morph one end vertex of an edge performs $O(1)$ rotations around the other end vertex.} We then consider two drawings $\Gamma_s$ and $\Gamma_t$ of an $n$-vertex path $P$, where $\Gamma_s$ lies on a straight-line, whereas $\Gamma_t$ has a spiral-like shape, and we prove that in any planar morph between $\Gamma_s$ and $\Gamma_t$ there is one edge of $P$ whose total rotation describes an angle whose size is $\Omega(n)$\remove{ that performs $\Omega(n)$ rotations}.

The rest of the paper is organized as follows. In Section~\ref{se:preliminaries} we give some definitions and preliminaries; in Section~\ref{se:algorithm} we present our algorithm; in Section~\ref{se:lower} we\remove{present our} discuss the lower bound; finally, in Section~\ref{se:conclusions} we conclude and\remove{present} offer some open problems.

\section{Preliminaries} \label{se:preliminaries}

In this section we give some definitions and preliminaries.

\vspace{-0.05cm}
\subsection{Drawings and Embeddings}

A \emph{planar straight-line drawing} of a graph maps each vertex to a distinct point in the plane and each edge to a straight-line segment between its endpoints so that no two edges cross. A planar drawing partitions the plane into topologically connected regions, called  {\em faces}. The bounded faces are \emph{internal}, while the unbounded face is the \emph{outer face}. A planar straight-line drawing is {\em convex} if each face is delimited by a convex polygon. A planar drawing of a graph determines a circular ordering of the edges incident to each vertex, called {\em rotation system}. Two drawings of a graph are \emph{equivalent} if they have the same rotation system and the same outer face. A \emph{plane embedding} is an equivalence class of planar drawings. A graph with a plane embedding is called a \emph{plane graph}. A plane graph is \emph{maximal} if no edge can be added to it while maintaining its planarity.

\subsection{Subgraphs and Connectivity}

A {\em subgraph} $G'(V',E')$ of a graph $G(V,E)$ is a graph such that $V'\subseteq V$ and $E'\subseteq E$; $G'$ is {\em induced} if, for every $u,v \in V'$, $(u,v)\in E'$ if and only if $(u,v)\in E$. If $G$ is a plane graph, then a subgraph $G'$ of $G$ is regarded as a plane graph whose plane embedding is the one obtained from $G$ by removing all the vertices and edges not in $G'$.

A graph $G$ is \emph{connected} if there is a path between every pair of vertices; it is \emph{$k$-connected} if removing any $k-1$ vertices leaves $G$ connected; a {\em separating $k$-set} is a set of $k$ vertices whose removal disconnects $G$. A {\em separating $3$-cycle} in a plane graph $G$ is a cycle with three vertices containing vertices both in its interior and in its exterior. Every separating $3$-set in a maximal plane graph $G$ induces a separating $3$-cycle.

\subsection{Monotonicity} \label{subse:monotonicity}

An {\em arc} $\vec{xy}$ is a line segment having $x$ and $y$ as endpoints and directed from $x$ to $y$. An arc $\vec{xy}$ is \emph{monotone} with respect to an oriented straight line $\vec d$ if it has a \emph{positive projection} on $\vec d$. That is, let $p$ and $q$ be any two distinct points in this order along $\vec {xy}$ when traversing $\vec{xy}$ from $x$ to $y$; then, the projection of $p$ on $\vec d$ precedes the projection of $q$ on $\vec d$ when traversing $\vec d$ according to its orientation. A path $P= (u_1, \dots, u_n)$ is \emph{$\vec d$-monotone} if the straight-line arc $\vec{u_i u_{i+1}}$ is monotone with respect to $\vec d$, for $i=1, \dots, n-1$; a path $P$ is \emph{monotone} if there exists an oriented straight line $\vec d$ such that $P$ is $\vec d$-monotone. A polygon $Q$ is \emph{$\vec d$-monotone} if there exist two vertices $s$ and $t$ in $Q$ such that the two paths that start at $s$, that end at $t$, and that compose $Q$ are both $\vec d$-monotone. Finally, a polygon $Q$ is {\em monotone} if there exists an oriented straight line $\vec d$ such that $Q$ is $\vec d$-monotone. We show some lemmata about monotone paths and polygons.

\begin{lemma}\label{le:convex-is-monotone}
Let $Q$ be any convex polygon and let $\vec d$ be any oriented straight line not perpendicular to any straight line through two vertices of $Q$. Then $Q$ is $\vec d$-monotone.
\end{lemma}
\begin{proof}
Refer to Fig.~\ref{fig:convex-is-monotone}. Denote by $u_1,\dots,u_k$ the vertices of $Q$, in any order. Let $\vec d$ be any oriented straight line not perpendicular to any straight line through two vertices of $Q$. For $1\leq i\leq k$, let $u'_i$ be the projection of $u_i$ on $\vec d$. Since $Q$ is convex and $\vec d$ is not perpendicular to any straight line through two vertices of $Q$, we have that $u'_i$ and $u'_j$ are distinct, for $1\leq i\neq j \leq k$. Let $\sigma$ be the total order of the projections $u'_i$ as they are encountered when traversing $\vec d$ according to its orientation. Let $u'_a$ and $u'_b$ be the first and the last element in $\sigma$, respectively. We claim that the two paths $P_1$ and $P_2$ connecting $u_a$ and $u_b$ along $Q$ are $\vec d$-monotone. The claim directly implies the lemma.

\begin{figure}[htb]
\begin{center}
\mbox{\includegraphics[scale=0.45]{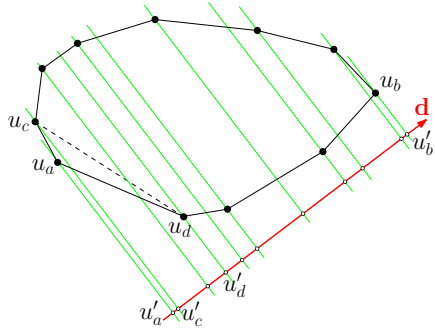}}
\caption{Illustration for the proof of Lemma~\ref{le:convex-is-monotone}.}
\label{fig:convex-is-monotone}
\end{center}
\end{figure}

We prove the claim by induction on $k$. If $k=3$, then the claim is trivially proved. If $k \geq 4$, then let $u'_c$ be the second element in $\sigma$. Assume, w.l.o.g., that $u_c$ is in $P_1$. Then, let $Q'$ be the convex polygon obtained from $Q$ by inserting a segment connecting $u_c$ with the second vertex of $P_2$, say $u_d$, and by removing $u_a$ and its two incident segments. Let $\sigma'=\sigma\setminus\{u'_a\}$. By assumption, $u'_c$ and $u'_b$ are the first and the last element in $\sigma'$, respectively. By induction, the two paths $P'_1=P_1\setminus\{(u_a,u_c)\}$ and $P'_2=P_2\setminus\{(u_a,u_d)\}\cup\{u_c,u_d\}$ are $\vec d$-monotone. Finally, arcs $\vec{u_a u_c}$ and $\vec{u_a u_d}$ have positive projections on $\vec d$, by the assumption that $u'_a$ is the first element in $\sigma$. Hence, paths $P_1$ and $P_2$ are $\vec d$-monotone and polygon $Q$ is $\vec d$-monotone.
\end{proof}

\begin{lemma} \label{le:angles-monotone-path}
Let $P=(u_1,u_2,u_3,u_4)$ be a path drawn in the plane. Denote by $\alpha$ the angle spanned by segment $\overline{u_1 u_2}$ while rotating such a segment clockwise around $u_2$ until it overlaps segment $\overline{u_2 u_3}$. Also, denote by $\beta$ the angle spanned by segment $\overline{u_2 u_3}$ while rotating such a segment clockwise around $u_3$ until it overlaps segment $\overline{u_3 u_4}$. Then, $P$ is monotone if and only if $\pi < \alpha + \beta < 3\pi$.
\end{lemma}
\begin{proof}
Let $\alpha'=2\pi-\alpha$ and $\beta'=2\pi-\beta$ be the two angles incident to $u_2$ and to $u_3$ different from $\alpha$ and from $\beta$, respectively. Observe that $\pi < \alpha' + \beta' < 3\pi$ if and only if $\pi < \alpha + \beta < 3\pi$.

First, suppose that $P$ is monotone, that is, there exists an oriented straight line $\vec d$ such that $P$ is $\vec d$-monotone. We prove that $\pi < \alpha + \beta < 3\pi$. Refer to Fig.~\ref{fig:path}(a). Denote by $u'_1$ and $u'_4$ the projections of $u_1$ and $u_4$ on $\vec d$, respectively. Consider polygon $Q=(u_1,u_2,u_3,u_4,u'_4,u'_1)$. Denote by $\delta_1$, $\delta_4$, $\delta'_1$, and $\delta'_4$ the angles incident to $u_1$, $u_4$, $u'_1$, and $u'_4$ inside $Q$, respectively. We have $\alpha+\beta+\delta_1+\delta_4+\delta'_1+\delta'_4=4\pi$. Further, $\delta'_1=\delta'_4=\pi/2$. Since $0<\delta_1,\delta_4<\pi$, it follows that $\pi<\alpha+\beta<3\pi$.


Second, suppose that $\pi < \alpha + \beta < 3\pi$. We prove that $P$ is monotone. We assume that $\alpha + \beta \leq 2\pi$. Indeed, if $\alpha + \beta > 2\pi$, then $\alpha' + \beta' \leq 2\pi$ and a symmetric proof can be exhibited in which $\alpha'$ and $\beta'$ replace $\alpha$ and $\beta$. Also, assume that $\alpha\leq \beta$, as the case $\beta\leq \alpha$ can be dealt with symmetrically.

\begin{figure}
    \centering
  \begin{minipage}{0.2\textwidth}
    \centering \vfill
    \subfigure[]{\includegraphics[scale=0.4]{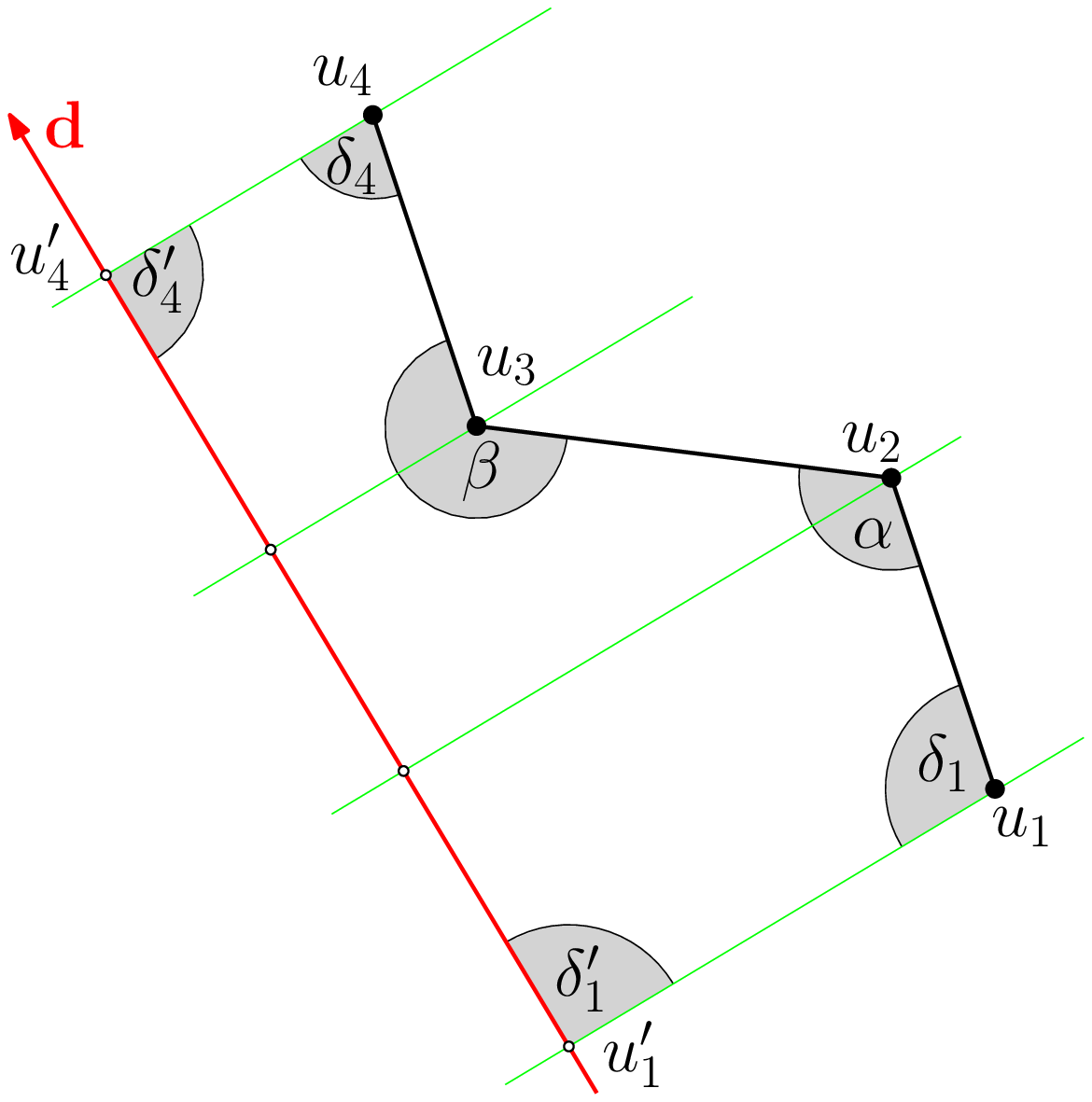}}
  \end{minipage}\hspace{2cm}
  \begin{minipage}{0.5\textwidth}
    \centering
\subfigure[]{\includegraphics[scale=0.4]{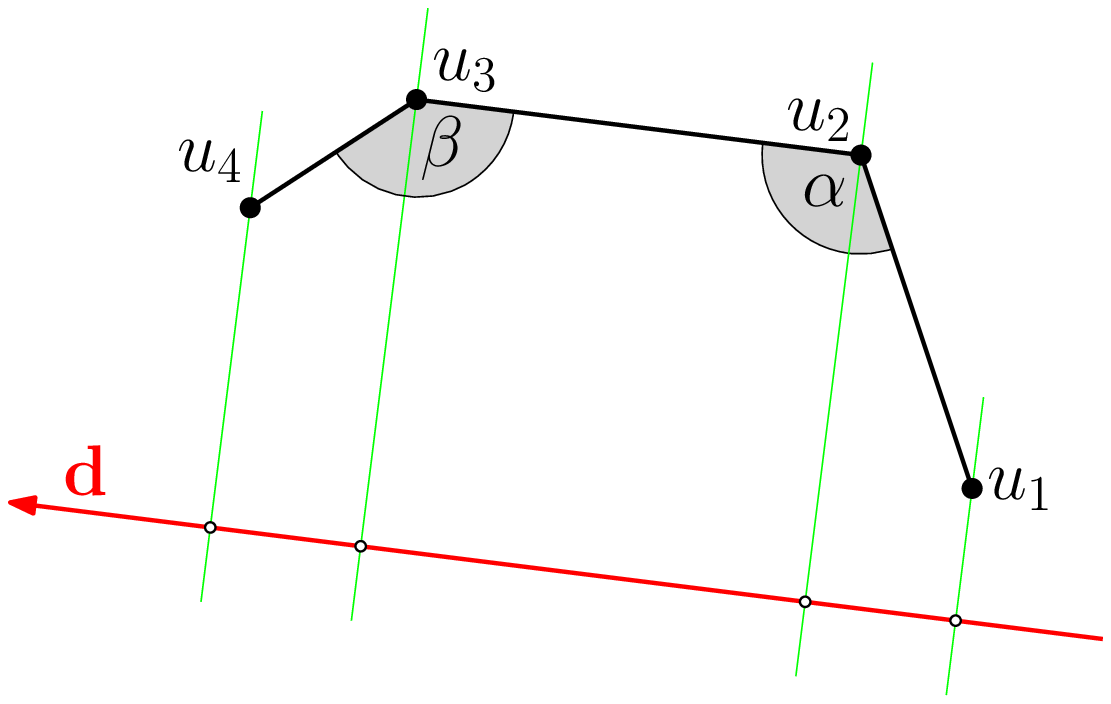}}

\subfigure[]{\includegraphics[scale=0.4]{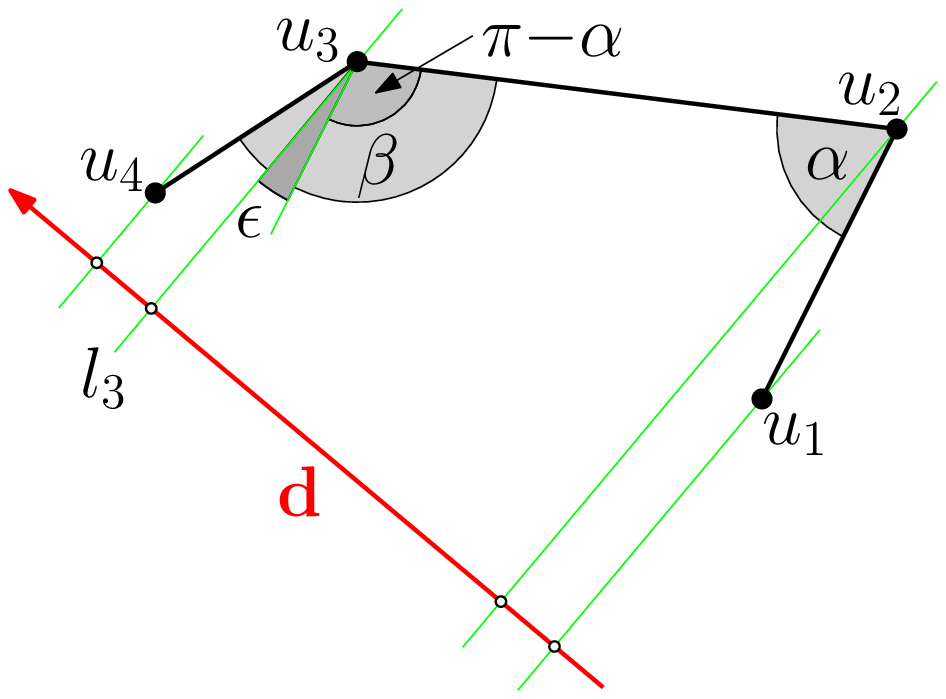}}  \end{minipage}
\caption{(a) If $P$ is monotone, then $\pi < \alpha + \beta < 3\pi$. (b) If $\pi < \alpha + \beta < 3\pi$ and $\alpha>\pi/2$, then $P$ is monotone. (c) If $\pi < \alpha + \beta < 3\pi$ and $\alpha \leq \pi/2$, then $P$ is monotone. }
\label{fig:path}
\end{figure}

If $\alpha>\pi/2$, then $\pi/2<\beta<3\pi/2$. Refer to Fig.~\ref{fig:path}(b). Let $\vec d$ be the oriented straight line parallel to segment $\overline{u_2 u_3}$ and oriented in such a way that arc $\vec{u_2 u_3}$ has a positive projection on $\vec d$. Since $\alpha,\beta>\pi/2$ and since $\alpha,\beta<3\pi/2$, it follows that arcs $\vec{u_1 u_2}$ and $\vec{u_3 u_4}$ have a positive projection on $\vec d$ as well, hence $P$ is $\vec d$-monotone.

If $\alpha\leq \pi/2$, then let $\epsilon$ be an arbitrarily small positive value such that $\beta>\pi-\alpha+\epsilon$. Such an $\epsilon$ always exist, given that $\beta>\pi-\alpha$. Refer to Fig.~\ref{fig:path}(c). Let $l_3$ be the line through $u_3$ such that the angle spanned by $\overline{u_2 u_3}$ while clockwise rotating such a segment around $u_3$ until it overlaps $l_3$ is equal to $\pi-\alpha+\epsilon$. Let $\vec d$ be an oriented straight line orthogonal to $l_3$ and directed so that arc $\vec{u_2u_3}$ has a positive projection on it. Observe that segment $\overline{u_2 u_3}$ is not perpendicular to $\vec d$, given that $\overline{u_2 u_3}$ and $l_3$ form an angle of $\pi-\alpha+\epsilon<\pi$. We claim that $P$ is $\vec d$-monotone. Arc $\vec{u_2u_3}$ has a positive projection on $\vec d$ by construction. The angle spanned by a clockwise rotation of segment $\overline{u_1 u_2}$ around $u_2$ bringing $\overline{u_1 u_2}$ to overlap with a line orthogonal to $\vec d$ passing through $u_2$ is $\epsilon$ by construction. Hence, arc $\vec{u_1u_2}$ has a positive projection on $\vec d$, given that $0<\epsilon<\pi$. Finally, to prove that arc $\vec{u_3u_4}$ has a positive projection on $\vec d$, it suffices to observe that $u_4$ is in the half-plane delimited by $l_3$ and not containing $u_2$, as a consequence of $\beta>\pi-\alpha+\epsilon$ and $\alpha + \beta \leq 2\pi$.

This concludes the proof of the lemma.
\end{proof}

\begin{lemma} \label{le:monotone-in-one-direction}
Any planar polygon $Q$ with at most $5$ vertices is monotone.
\end{lemma}
\begin{proof}
The proof distinguishes three cases, depending on the number of vertices of~$Q$.
\begin{itemize}
\item If $Q$ has three vertices, then it is convex, hence the statement follows from Lemma~\ref{le:convex-is-monotone}.

\item If $Q$ has four vertices, then it suffices to show that $Q$ contains a monotone path with four vertices. Namely, assume that $Q$ contains a path $P=(u_1,u_2,u_3,u_4)$ which is monotone with respect to some oriented straight line $\vec d$. Then, paths $(u_1,u_2,u_3,u_4)$ and $(u_1,u_4)$ are both $\vec d$-monotone, hence $Q$ is $\vec d$-monotone.

    Denote by $\alpha$, $\beta$, $\gamma$, and $\delta$ the angles internal to $Q$ in clockwise order around $Q$. Since $\alpha + \beta +  \gamma + \delta= 2\pi$, it follows that $\alpha + \beta < 3\pi$, that $\beta + \gamma < 3\pi$, that $\gamma + \delta < 3\pi$, and that $\delta + \alpha < 3\pi$. Suppose that for two consecutive angles in $Q$, say $\alpha$ and $\beta$, it holds $\alpha + \beta < \pi$; then, $\pi < \gamma + \delta < 3\pi$, and hence $Q$ contains a monotone path with four vertices by Lemma~\ref{le:angles-monotone-path}. Thus, if $Q$ does not contain any monotone path with four vertices, then every two consecutive angles in $Q$ sum up to exactly $\pi$, hence $Q$ is convex, and it is monotone with respect to every oriented straight line $\vec d$ by Lemma~\ref{le:convex-is-monotone}.

\item If $Q$ has five vertices, then again it suffices to show that $Q$ contains a monotone path with four vertices. Namely, assume that $Q$ contains a monotone path $P=(u_1,u_2,u_3,u_4)$. By definition of monotone path, there exists an oriented straight line $\vec d$ such that arcs $\vec{u_1 u_2}$, $\vec{u_2 u_3}$, and $\vec{u_3 u_4}$ have positive projections on $\vec d$. Slightly perturb the slope of $\vec d$, if necessary, so that no line through two vertices of $Q$ is orthogonal to $\vec d$. If the perturbation is small enough, then $P$ is still $\vec d$-monotone. Denote by $u_5$ the fifth vertex of $Q$ and, for $1\leq i\leq 5$, denote by $u'_i$ the projection of $u_i$ on $\vec d$. If $u'_5$ precedes $u'_1$ on $\vec d$, then paths $(u_5,u_1,u_2,u_3,u_4)$ and $(u_5,u_4)$ are both $\vec d$-monotone, hence $Q$ is $\vec d$-monotone; if $u'_5$ follows $u'_4$ on $\vec d$, then paths $(u_1,u_2,u_3,u_4,u_5)$ and $(u_1,u_5)$ are both $\vec d$-monotone, hence $Q$ is $\vec d$-monotone; finally, if $u'_5$ follows $u'_1$ and precedes $u'_4$ on $\vec d$, then paths $(u_1,u_2,u_3,u_4)$ and $(u_1,u_5,u_4)$ are both $\vec d$-monotone, hence $Q$ is $\vec d$-monotone.

    Denote by $\alpha$, $\beta$, $\gamma$, $\delta$, and $\epsilon$ the angles internal to $Q$ in clockwise order around $Q$. Since $\alpha + \beta +  \gamma + \delta + \epsilon= 3\pi$, it follows that $\alpha + \beta < 3\pi$, that $\beta + \gamma < 3\pi$, that $\gamma + \delta < 3\pi$, that $\delta + \epsilon < 3\pi$, and that $\epsilon + \alpha < 3\pi$. Suppose next that $\alpha + \beta \leq \pi$, that $\beta + \gamma \leq \pi$, that $\gamma + \delta \leq \pi$, that $\delta + \epsilon \leq \pi$, and that $\epsilon + \alpha \leq \pi$. Summing up the inequalities gives $2\alpha + 2\beta + 2 \gamma + 2\delta + 2\epsilon\leq 5\pi$, hence $\alpha + \beta +  \gamma + \delta + \epsilon\leq 5\pi/2$, a contradiction to the fact that $\alpha + \beta +  \gamma + \delta + \epsilon= 3\pi$. Hence, for at least a pair of consecutive angles of $Q$, say $\alpha$ and $\beta$, it holds $\pi < \alpha + \beta < 3\pi$. Thus, by Lemma~\ref{le:angles-monotone-path}, $Q$ contains a monotone path with four vertices.
\end{itemize}
This concludes the proof of the lemma.
\end{proof}

\subsection{Morphing}

A \emph{linear morph} \mmorph{\Gamma_1,\Gamma_2} is a continuous transformation between two straight-line planar drawings $\Gamma_1$ and $\Gamma_2$ of a plane graph $G$ such that each vertex moves at constant speed along a straight line from its position in $\Gamma_1$ to the one in $\Gamma_2$. A linear morph is {\em planar} if no crossing or overlap occurs between any two edges or vertices during the transformation. A planar linear morph is also called a {\em morphing step}. In the remainder of the paper, we will construct {\em unidirectional} linear morphs, that were defined in~\cite{bhl-mpgdum-13} as linear morphs in which the straight-line trajectories of the vertices are parallel.

A \emph{morph} \mmorph{\Gamma_s,\dots,\Gamma_t} between two straight-line planar drawings $\Gamma_s$ and $\Gamma_t$ of a plane graph $G$ is a finite sequence of morphing steps that transforms $\Gamma_s$ into $\Gamma_t$. A {\em unidirectional morph} is such that each of its morphing steps is unidirectional.

Let $\Gamma$ be a planar straight-line drawing of a plane graph $G$. The \emph{kernel} of a vertex $v$ of $G$ in $\Gamma$ is the open convex region $R$ such that placing $v$ at any point of $R$ while maintaining unchanged the position of every other vertex of $G$ yields a planar straight-line drawing of $G$. If a neighbor $x$ of $v$ lies on the boundary of the kernel of $v$ in $\Gamma$, we say that $v$ is \emph{$x$-contractible}. The \emph{contraction of $v$ onto $x$} in $\Gamma$ is the operation resulting in: (i) a simple graph $G'= G/(v,x)$ obtained from $G$ by removing $v$ and by replacing each edge $(v,w)$, where $w \neq x$, with an edge $(x,w)$ (if it does not already belong to $G$); and (ii) a planar straight-line drawing $\Gamma'$ of $G'$ such that each vertex different from $v$ is mapped to the same point as in $\Gamma$. Also, the {\em uncontraction of $v$ from $x$ into $\Gamma$} is the reverse operation of the contraction of $v$ onto $x$ in $\Gamma$, i.e., the operation that produces a planar straight-line drawing $\Gamma$ of $G$ from a planar straight-line drawing $\Gamma'$ of $G'$.

A vertex $v$ in a plane graph $G$ is a {\em quasi-contractible vertex} if (i) $\textrm{deg}(v)\le 5$ and, (ii) for any two neighbors $u$ and $w$ of $v$, if $u$ and $w$ are adjacent, then $(u,v,w)$ is a face of $G$. We have the following.

\begin{lemma}\label{le:candidate_exists}(Angelini et al.~\cite{afpr-mpgde-13})
Every plane graph contains a quasi-contractible vertex.
\end{lemma}

In the remainder of the paper, even when not explicitly specified, we will only consider and perform contractions of quasi-contractible vertices.

Let $\Gamma_1$ and $\Gamma_2$ be two straight-line planar drawings of the same plane graph $G$. We define a \emph{pseudo-morph} of $\Gamma_1$ into $\Gamma_2$ as follows: (A) a unidirectional morph with $m$ morphing steps of $\Gamma_1$ into $\Gamma_2$ is a pseudo-morph with $m$ steps of $\Gamma_1$ into $\Gamma_2$; (B) a unidirectional morph with $m_1$ morphing steps of $\Gamma_1$ into a straight-line planar drawing $\Gamma^x_1$ of $G$, followed by a pseudo-morph with $m_2$ steps of $\Gamma^x_1$ into a straight-line planar drawing $\Gamma^x_2$ of $G$, followed by a unidirectional morph with $m_3$ morphing steps of $\Gamma^x_2$ into $\Gamma_2$ is a pseudo-morph of $\Gamma_1$ into $\Gamma_2$ with $m_1+m_2+m_3$ steps; and (C) denote by $\Gamma'_1$ and $\Gamma'_2$ the straight-line planar drawings of the plane graph $G'$ obtained by contracting a quasi-contractible vertex $v$ of $G$ onto $x$ in $\Gamma_1$ and in $\Gamma_2$, respectively; then, the contraction of $v$ onto $x$, followed by a pseudo-morph with $x$ steps of $\Gamma'_1$ into $\Gamma'_2$, followed by the uncontraction of $v$ from $x$ into $\Gamma_2$ is a pseudo-morph with $m+2$ steps of $\Gamma_1$ into $\Gamma_2$.

Pseudo-morphs have two useful and powerful features.

First, it is easy to design an inductive algorithm for constructing a pseudo-morph between any two planar straight-line drawings $\Gamma_1$ and $\Gamma_2$ of the same $n$-vertex plane graph $G$. Namely, consider any quasi-contractible vertex $v$ of $G$ and let $x$ be any neighbor of $v$. Morph unidirectionally $\Gamma_1$ and $\Gamma_2$ into two planar straight-line drawings $\Gamma^x_1$ and $\Gamma^x_2$, respectively, in which $v$ is $x$-contractible. Now contract $v$ onto $x$ in $\Gamma^x_1$ and in $\Gamma^x_2$ obtaining two planar straight-line drawings  $\Gamma_1'$ and $\Gamma_2'$, respectively, of the same $(n-1)$-vertex plane graph $G'$. Then, the algorithm is completed by inductively computing a pseudo-morph of $\Gamma_1'$ into $\Gamma_2'$.

Second, computing a pseudo-morph between $\Gamma_1$ and $\Gamma_2$ leads to computing a planar unidirectional morph between $\Gamma_1$ and $\Gamma_2$, as formalized in Lemma~\ref{le:pseudo-to-morph}. We remark that, although Lemma~\ref{le:pseudo-to-morph} has never been stated as below, its proof can be directly derived from the results of Alamdari et al.~\cite{aac-mpgdpns-13,afpr-mpgde-13} and, mainly, of Barrera-Cruz et al.~\cite{bhl-mpgdum-13}.

\begin{lemma} \label{le:pseudo-to-morph} 
Let $\Gamma_s$ and $\Gamma_t$ be two straight-line planar drawings of a plane graph $G$. Let $\cal P$ be a pseudo-morph with $m$ steps transforming $\Gamma_s$ into $\Gamma_t$. It is possible to construct a planar unidirectional morph $M$ with $m$ morphing steps transforming $\Gamma_s$ into $\Gamma_t$.
\end{lemma}
\begin{proof}
The proof is by induction primarily on the number $k$ of contractions in $\cal P$ and secondarily on the number $x$ of steps of $\cal P$.

If $k=0$, then we are in Case (A) of the definition of a pseudo-morph; hence, $\cal P$ is a planar unidirectional morph with $x$ morphing steps transforming $\Gamma_s$ into $\Gamma_t$.

If $k>0$ and the first step of $\cal P$ is a unidirectional morphing step transforming $\Gamma_s$ into a straight-line planar drawing $\Gamma'_s$ of $G$, then we are in Case (B) of the definition of a pseudo-morph; denote by ${\cal P}'$ the pseudo-morph composed of the last $m-1$ steps of $\cal P$. By induction, there exists a planar unidirectional morph $M'$ with $m-1$ morphing steps transforming $\Gamma'_s$ into $\Gamma_t$. Hence the first morphing step of $\cal P$  followed by $M'$ is a planar unidirectional morph with $x$ morphing steps transforming $\Gamma_s$ into $\Gamma_t$.

The case in which $k>0$ and the last step of $\cal P$ is a unidirectional morphing step can be discussed analogously.

If $k>0$ and neither the first nor the last step of $\cal P$ is a unidirectional morphing step, then we are in Case (C) of the definition of a pseudo-morph. Hence, the first step of $\cal P$ is a contraction of a quasi-contractible vertex $v$ on a neighbor $x$, resulting in a planar straight-line drawing  $\Gamma_s'$ of an $(n-1)$-vertex plane graph $G'$. Also, the last step of $\cal P$ starts from a drawing $\Gamma_t'$ of $G'$ and uncontracts $v$ from $x$ into $\Gamma_t$.

Denote by ${\cal P}'$ the pseudo-morph with $m-2$ steps that is the part of $\cal P$ transforming $\Gamma_s'$ into $\Gamma_t'$. By induction, there exists a planar unidirectional morph $M'=\morph{\Gamma'_s = \Gamma'_1, \dots, \Gamma'_{m-2}=\Gamma'_t}$ with $m-2$ morphing steps transforming $\Gamma'_s$ into $\Gamma'_t$. For each $i=1,\dots,x-2$, we are going to construct a drawing $\Gamma_i$ of $G$ by placing vertex $v$ in a suitable position in $\Gamma'_i$ in such a way that the morph $M$ with $m$ morphing steps composed of a morphing step \mmorph{\Gamma_s, \Gamma_1}, followed by the morph \mmorph{\Gamma_1, \dots, \Gamma_{m-2}}, followed by a morphing step \mmorph{\Gamma_{m-2},\Gamma_t} is planar and unidirectional.

This strategy of constructing $M$ starting from $M'$ by suitably placing $v$ in each drawing of $M'$ is the same that was applied in~\cite{aac-mpgdpns-13,afpr-mpgde-13,bhl-mpgdum-13}. It should be noted that the algorithm for placing $v$ in $\Gamma'_1, \dots, \Gamma'_{x-2}$ differs slightly in those three papers. We opt here for an algorithm almost identical to the one in~\cite{bhl-mpgdum-13}, as it ensures that $M$ is a unidirectional morph. However, since in~\cite{bhl-mpgdum-13} $G$ is assumed to be a maximal plane graph, vertex $v$ can always be chosen to be an internal vertex of $G$ with degree at least $3$. In our case, instead, $v$ might be incident to the outer face of $G$ and might have degree $1$ or $2$.

We now describe the algorithm in~\cite{bhl-mpgdum-13} for placing $v$ when $v$ is internal and $\textrm{deg}(v) = 5$; then, we will argue that an analogous technique can be applied even if $v$ is incident to the outer face of $G$ and has degree $1$ or $2$.

Observe that, at any time instant $t$ during $M'$, there exists a disk of radius $\epsilon_t > 0$ that is centered at $m$ and that does not contain any vertex or edge other than $x$ and its incident edges. Let $\epsilon = \min_t\{\epsilon_t\}$ be the minimum $\epsilon_t$ among all time instants $t$ of $M'$.

\begin{figure}[htb]
    \centering
    \subfigure[]{\includegraphics[scale=0.35]{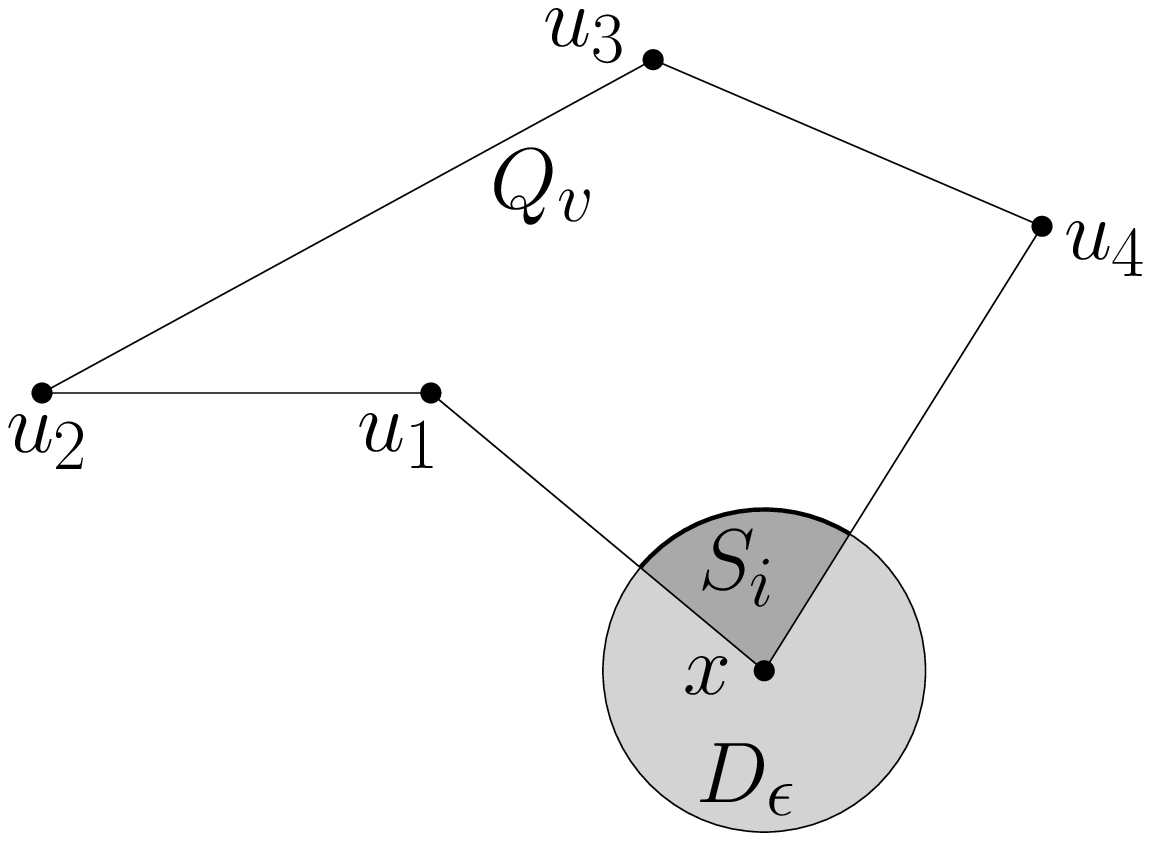}}\hspace{1cm}
    \subfigure[]{\includegraphics[scale=0.35]{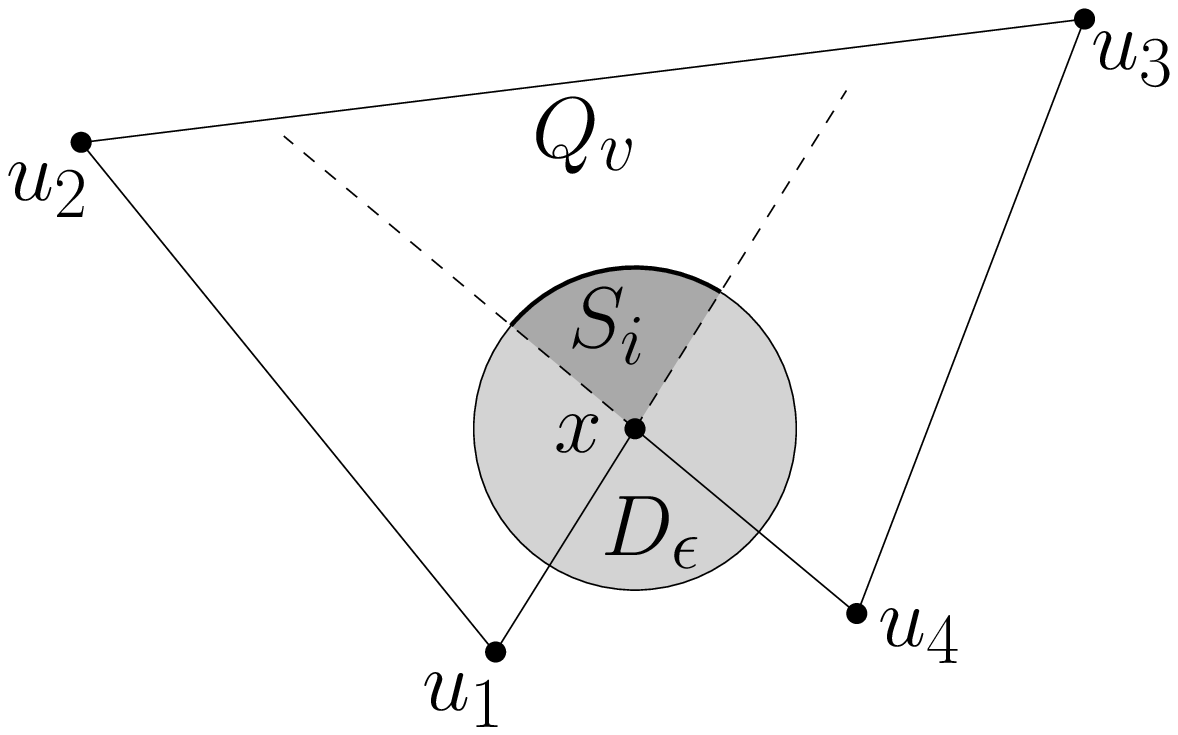}}
\caption{Circular sector $S_i$ if (a) the internal angle of $Q_v$ incident to $x$ is smaller than or equal to $\pi$ or (b) the internal angle of $Q_v$ incident to $x$ is larger than $\pi$.}
\label{fig:sector}
\end{figure}

For each $i=1,\dots,m-2$, let $S_i$ be the circular sector resulting from the intersection between a disk $D_{\epsilon}$ centered at $x$ with radius $\epsilon$ and the kernel of the polygon $Q_v$ induced by the neighbors of $v$ in $\Gamma_i$. In particular (see Fig.~\ref{fig:sector}), if the internal angle of $Q_v$ incident to $x$ is smaller than or equal to $\pi$, then $S_i$ is delimited by the two radii of $D_{\epsilon}$ that overlap with the two edges of $Q_v$ incident to $x$, while if such an angle is larger than $\pi$ then $S_i$ is delimited by the two radii of $D_{\epsilon}$ that overlap with the elongations emanating from $x$ of the two edges of $Q_v$ incident to $x$. Barrera-Cruz et al. prove in~\cite{bhl-mpgdum-13} that each circular sector $S_i$ contains at least one \emph{nice} point, defined as follows. All the points of $S_{m-2}$ are nice. For $i=1,\dots,m-3$, a point $p_i$ of $S_i$ is nice if there exists a nice point $p_{i+1}$ in $S_{i+1}$ such that the line passing through $p_{i}$ and $p_{i+1}$ is parallel to the trajectory followed by each vertex during the unidirectional morphing step transforming $\Gamma'_i$ into $\Gamma'_{i+1}$. The proof in~\cite{bhl-mpgdum-13} is completed by showing that placing $v$ on the nice point $p_{i}$ in $\Gamma'_i$ and on the corresponding nice point $p_{i+1}$ in $\Gamma'_{i+1}$ yields two drawings $\Gamma_{i}$ and $\Gamma_{i+1}$ of $G$ such that \mmorph{\Gamma_{i},\Gamma_{i+1}} is planar and, by construction, unidirectional.

In order to adapt this algorithm to our setting, it is sufficient to describe how to compute each circular sector $S_i$, since the rest of the proof works exactly as described in~\cite{bhl-mpgdum-13} for the case in which $\textrm{deg}(v) = 5$. The complication here is that the neighbors of $v$ might not create a polygon $Q_v$ enclosing $v$ in its interior, hence it is not possible to use the concept of ``kernel of a polygon'' in order to define $S_i$. To overcome this problem, we use the concept of ``kernel of a vertex'' $v$, defined as the region of the plane such that each of its points has direct visibility to all the neighbors of $v$. Observe that this is the same property satisfied by the kernel of $Q_v$, however the kernel of $v$ is well-defined even if the neighbors of $v$ do not induce a polygon enclosing $v$ in its interior, e.g., if $v$ is incident to the outer face or $\textrm{deg}(v)\le 2$.

More in detail, if $\textrm{deg}(v) = 1$, then $S_i$ is the intersection of $D_{\epsilon}$ with the region of $\Gamma'_i$ representing the face of $G'$ that contains $v$ in $G$. If $\textrm{deg}(v) = 2$, then $S_i$ is the intersection of: (i) $D_{\epsilon}$, (ii) the region of $\Gamma'_i$ representing the face of $G'$ that contains $v$ in $G$, and (iii) the half-plane that is to the left (right) of the oriented straight line from a neighbor $w$ of $v$ to the other neighbor $z$ of $v$ if $w$, $z$, and $v$ appear in this counter-clockwise (resp. clockwise) order along cycle $(w,z,v)$ in $G$. Finally, if $3 \le \textrm{deg}(v) \le 5$, then let $w$ and $z$ be the two neighbors of $x$ in $G$ such that edges $(x,w)$, $(x,v)$, and $(x,z)$ appear consecutively around $x$ in this clockwise order; then, if the angle spanned when rotating $(x,w)$ clockwise till coinciding with $(x,z)$ is smaller than or equal to $\pi$, then $S_i$ is delimited by the two radii of $D_{\epsilon}$ that overlap with edges $(x,w)$ and $(x,z)$, otherwise $S_i$ is delimited by the two radii of $D_{\epsilon}$ that overlap with the elongations of edges $(x,w)$ and $(x,z)$ emanating from $x$. We observe that an analogous definition of circular sectors $S_i$ was provided in~\cite{afpr-mpgde-13} (although the morphs constructed in~\cite{afpr-mpgde-13} are not unidirectional).

We conclude the proof by observing that the first and the last morphing steps \mmorph{\Gamma_s, \Gamma_1} and \mmorph{\Gamma_{m-2},\Gamma_t} of $M$ are planar, since $v$ has been placed on a nice point in $\Gamma_1$ and in $\Gamma_{m-2}$, and unidirectional, since $v$ is the only vertex moving during these two steps.
\end{proof}

\remove{
The basic ideas for a proof of Lemma~\ref{le:pseudo-to-morph}, which appear already with a different flavour in~\cite{aac-mpgdpns-13,afpr-mpgde-13}, are: (1) contractions and uncontractions can be simulated by single unidirectional morphing steps; and (2) consider a planar unidirectional morph $M'$ of the $(n-1)$-vertex plane graph $G'$ obtained from $G$ by contracting a quasi-contractible vertex $v$ onto a neighbor $x$; then, it is possible to design a suitable trajectory for $v$ so that $M'$ together with this trajectory provides a unidirectional morph $M$ of $\Gamma_s$ into $\Gamma_t$. See~\cite{aac-mpgdpns-13,afpr-mpgde-13,bhl-mpgdum-13} for further details.
}

\subsection{Hierarchical Graphs and Level Planarity} \label{subse:hierarchical-level}

A {\em hierarchical graph} is a tuple $(G,\vec d,L,\gamma)$ where: (i) $G$ is a graph; (ii) $\vec d$ is an oriented straight line in the plane; (iii) $L$ is a set of parallel lines (sometimes called {\em layers}) that are orthogonal to $\vec d$; the lines in $L$ are assumed to be ordered in the same order as they are intersected by $\vec d$ when traversing such a line according to its orientation; and (iv) $\gamma$ is a function that maps each vertex of $G$ to a line in $L$ in such a way that, if an edge $(u,v)$ belongs to $G$, then $\gamma(u)\neq \gamma(v)$. A {\em level drawing} of $(G,\vec d,L,\gamma)$ (sometimes also called {\em hierarchical drawing}) maps each vertex $v$ of $G$ to a point on the line $\gamma(v)$ and each edge $(u,v)$ of $G$ such that line $\gamma(u)$ precedes line $\gamma(v)$ in $L$ to an arc $\vec{uv}$ monotone with respect to $\vec d$. A {\em hierarchical plane graph} is a hierarchical graph $(G,\vec d,L,\gamma)$ such that $G$ is a plane graph and such that a level planar drawing $\Gamma$ of $(G,\vec d,L,\gamma)$ exists that ``respects'' the embedding of $G$ (that is, the rotation system and the outer face of $G$ in $\Gamma$ are the same as in the plane embedding of $G$). Given a hierarchical plane graph $(G,\vec d,L,\gamma)$, an {\em st-face} of $G$ is a face delimited by two paths $(s=u_1,u_2,\dots,u_k=t)$ and $(s=v_1,v_2,\dots,v_l=t)$ such that $\gamma(u_i)$ precedes $\gamma(u_{i+1})$ in $L$, for every $1\leq i\leq k-1$, and such that $\gamma(v_i)$ precedes $\gamma(v_{i+1})$ in $L$, for every $1\leq i\leq l-1$.  We say that $(G,\vec d,L,\gamma)$ is a {\em hierarchical plane st-graph} if every face of $G$ is an st-face. Let $\Gamma$ be any straight-line level planar drawing of a hierarchical plane graph $(G,\vec d,L,\gamma)$ and let $f$ be a face of $G$; then, it is easy to argue that $f$ is an st-face if and only if the polygon delimiting $f$ in $\Gamma$ is $\vec d$-monotone.

In this paper we will use a result of Hong and Nagamochi on the existence of convex straight-line level planar drawings of hierarchical plane st-graphs~\cite{hn-cdhpgcpg-10}. Here we explicitly formulate a weaker version of their main theorem.\footnote{We make some remarks. First, the main result in~\cite{hn-cdhpgcpg-10} proves that a convex straight-line level planar drawing of $(G,\vec d,L,\gamma)$ exists even if a convex polygon representing the cycle delimiting the outer face of $G$ is arbitrarily prescribed. Second, the result holds for a super-class of the triconnected planar graphs, namely for all the graphs that admit a convex straight-line drawing~\cite{cyn-lacdp-84,t-prg-84}. Third, the result assumes that the lines in $L$ are horizontal; however, a suitable rotation of the coordinate axes shows how that assumption is not necessary. Fourth, looking at the figures in~\cite{hn-cdhpgcpg-10} one might get the impression that the lines in $L$ need to be equidistant; however, this is nowhere used in their proof, hence the result holds for any set of parallel lines.}

\begin{theorem} \label{th:hong-nagamochi} (Hong and Nagamochi~\cite{hn-cdhpgcpg-10})
Let $(G,\vec d,L,\gamma)$ be a triconnected hierarchical plane st-graph. There exists a convex straight-line level planar drawing of $(G,\vec d,L,\gamma)$.
\end{theorem}

Let $\Gamma$ be a straight-line level planar drawing of a hierarchical plane graph $(G,\vec d,L,\gamma)$. Since each edge $(u,v)$ of $G$ is represented in $\Gamma$ by a $\vec d$-monotone arc, the fact that $(u,v)$ intersects a line $l_i\in L$ does not depend on the actual drawing $\Gamma$, but only on the fact that $l_i$ lies between lines $\gamma(u)$ and $\gamma(v)$ in $L$. Assume that each line $l_i\in L$ is oriented so that $\vec d$ cuts $l_i$ from the right to the left of $l_i$. We say that an edge $e$ {\em precedes} ({\em follows}) a vertex $v$ on a line $l_i$ in $\Gamma$ if $\gamma(v)=l_i$, $e$ intersects $l_i$ in a point $p_i(e)$, and $p_i(e)$ precedes (resp. follows) $v$ on $l_i$ when traversing such a line according to its orientation. Also, we say that an edge $e$ {\em precedes} ({\em follows}) an edge $e'$ on a line $l_i$ in $\Gamma$ if $e$ and $e'$ both intersect $l_i$ at points $p_i(e)$ and $p_i(e')$, and $p_i(e)$ precedes (resp. follows) $p_i(e')$ on $l_i$ when traversing such a line according to its orientation.

Now consider two straight-line level planar drawings $\Gamma_1$ and $\Gamma_2$ of a hierarchical plane graph $(G,\vec d,L,\gamma)$. We say that $\Gamma_1$ and $\Gamma_2$ are {\em left-to-right equivalent} if, for any line $l_i\in L$, for any vertex or edge $x$ of $G$, and for any vertex or edge $y$ of $G$, we have that $x$ precedes (follows) $y$ on $l_i$ in $\Gamma_1$ if and only if $x$ precedes (resp. follows) $y$ on $l_i$ in $\Gamma_2$. We are going to make use of the following lemma.

\begin{lemma} \label{le:unidirectional}
Let $\Gamma_1$ and $\Gamma_2$ be two left-to-right equivalent straight-line level planar drawings of the same hierarchical plane graph $(G,\vec d,L,\gamma)$. Then the linear morph \mmorph{\Gamma_1,\Gamma_2} transforming $\Gamma_1$ into $\Gamma_2$ is planar and unidirectional.
\end{lemma}

In order to prove Lemma~\ref{le:unidirectional}, we first recall an auxiliary lemma appeared in~\cite{bhl-mpgdum-13} stating that if two points $x$ and $y$ move at constant speed on the same line $l$ and $x$ precedes (follows) $y$ on $l$ both at the beginning and at the end of the movement, then $x$ precedes (follows) $y$ on $l$ during the whole movement.

\begin{lemma} \label{le:auxiliary-unidirectional}(Barrera-Cruz et al.~\cite{bhl-mpgdum-13})
Let $l$ be an oriented straight line and let $x_0$, $x_1$, $y_0$, and $y_1$ be points on $l$. Assume that $x_i$ precedes $y_i$ on $l$, for $i = 0,1$. Consider a point $x$ that moves in one unit of time from $x_0$ to $x_1$, and a point $y$ that moves in one unit of time from $y_0$ to $y_1$. Then, $x$ precedes $y$ on $l$ during the entire movement.
\end{lemma}

We now exhibit a proof of Lemma~\ref{le:unidirectional}.

\noindent\begin{proofx}
Morph \mmorph{\Gamma_1,\Gamma_2} is clearly unidirectional. We prove that it is planar.

Lemma~\ref{le:auxiliary-unidirectional} and the fact that $\Gamma_1$ and $\Gamma_2$ are left-to-right equivalent directly imply that, if two vertices lie on the same line $l \in L$, then they never overlap during \mmorph{\Gamma_1,\Gamma_2}.

We prove that there exists no overlap between a vertex $u$ and an edge $e$ of $G$ during \mmorph{\Gamma_1,\Gamma_2}. Such a proof also implies that there is no crossing between two edges at any time $t$ during \mmorph{\Gamma_1,\Gamma_2}; in fact, such a crossing can only happen if an end-vertex of one of the two edges overlaps the other edge at a time instant $t' \leq t$.

In order to prove that there exists no overlap between $u$ and $e$, it suffices to prove that the point $p_i(e)$ in which $e$ intersects line $l_i = \gamma(u)$ moves at constant speed during \mmorph{\Gamma_1,\Gamma_2}, since in this case Lemma~\ref{le:auxiliary-unidirectional} and the fact that $\Gamma_1$ and $\Gamma_2$ are left-to-right equivalent imply that $u$ and $p_i(e)$ never overlap.

The fact that $p_i(e)$ moves at constant speed during \mmorph{\Gamma_1,\Gamma_2} directly follows from: (i) the two end-vertices $v$ and $w$ of $e$ move at constant speed on two lines $\gamma(v)$ and $\gamma(w)$ that are parallel to $l_i$; and (ii) for any time instant $t$ of \mmorph{\Gamma_1,\Gamma_2}, the coefficients that express $p_i(e)$ as a convex combination of the positions of $v$ and $w$ are the same.


This concludes the proof of the lemma.
\end{proofx}


\section{A Morphing Algorithm} \label{se:algorithm}

In this section we describe an algorithm to construct a planar unidirectional morph with $O(n)$ steps between any two straight-line planar drawings $\Gamma_s$ and $\Gamma_t$ of the same $n$-vertex plane graph $G$. The algorithm relies on two subroutines, called {\sc fast convexifier} and {\sc contractibility creator}, which are described in Sections~\ref{subse:convexifier} and~\ref{subse:contractibility}, respectively. The algorithm is described in Section~\ref{subse:algorithm}.

\subsection{Fast Convexifier} \label{subse:convexifier}

Consider a straight-line planar drawing $\Gamma$ of an $n$-vertex maximal plane graph $G$, for some $n\geq 4$. Let $v$ be a quasi-contractible internal vertex of $G$ and let $C_v$ be the cycle of $G$ induced by the neighbors of $v$. See Fig.~\ref{fig:convexifier-process}(a). In this section we show an algorithm, that we call {\sc fast convexifier}, morphing $\Gamma$ into a straight-line planar drawing $\Gamma_M$ of $G$ in which $C_v$ is convex. Algorithm {\sc fast convexifier} consists of a single unidirectional morphing step.

\begin{figure}[htb]
\begin{center}
\begin{tabular}{c c c}
\mbox{\includegraphics[scale=0.278]{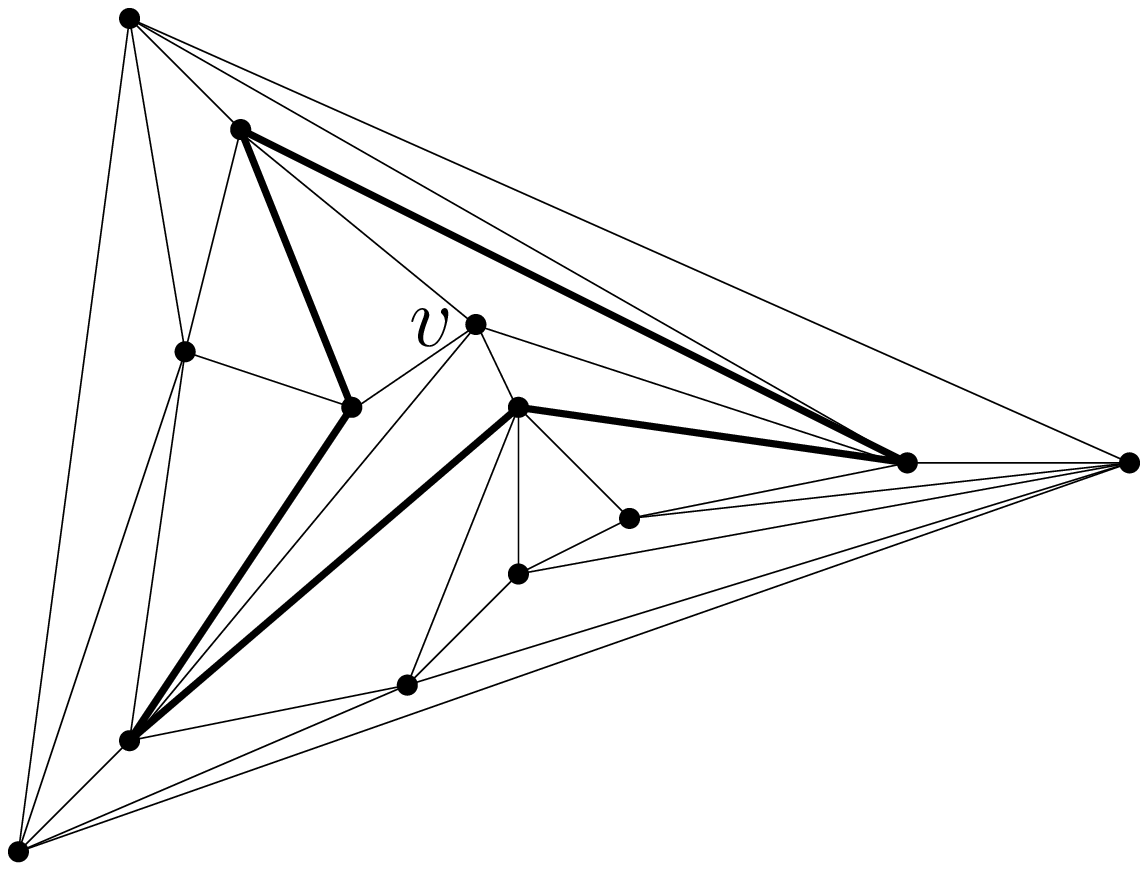}} \hspace{1mm} &
\mbox{\includegraphics[scale=0.278]{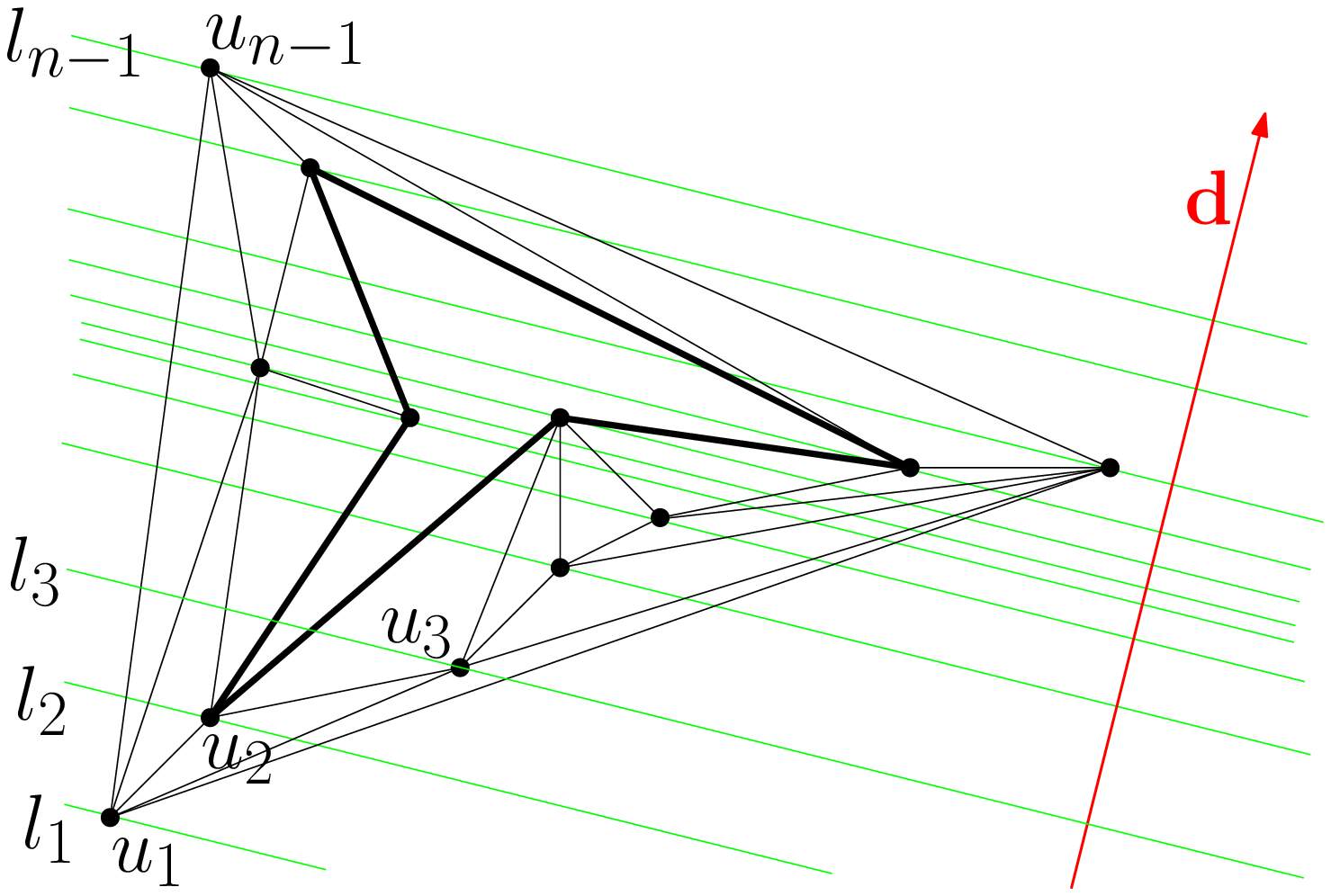}} \hspace{1mm} &
\mbox{\includegraphics[scale=0.278]{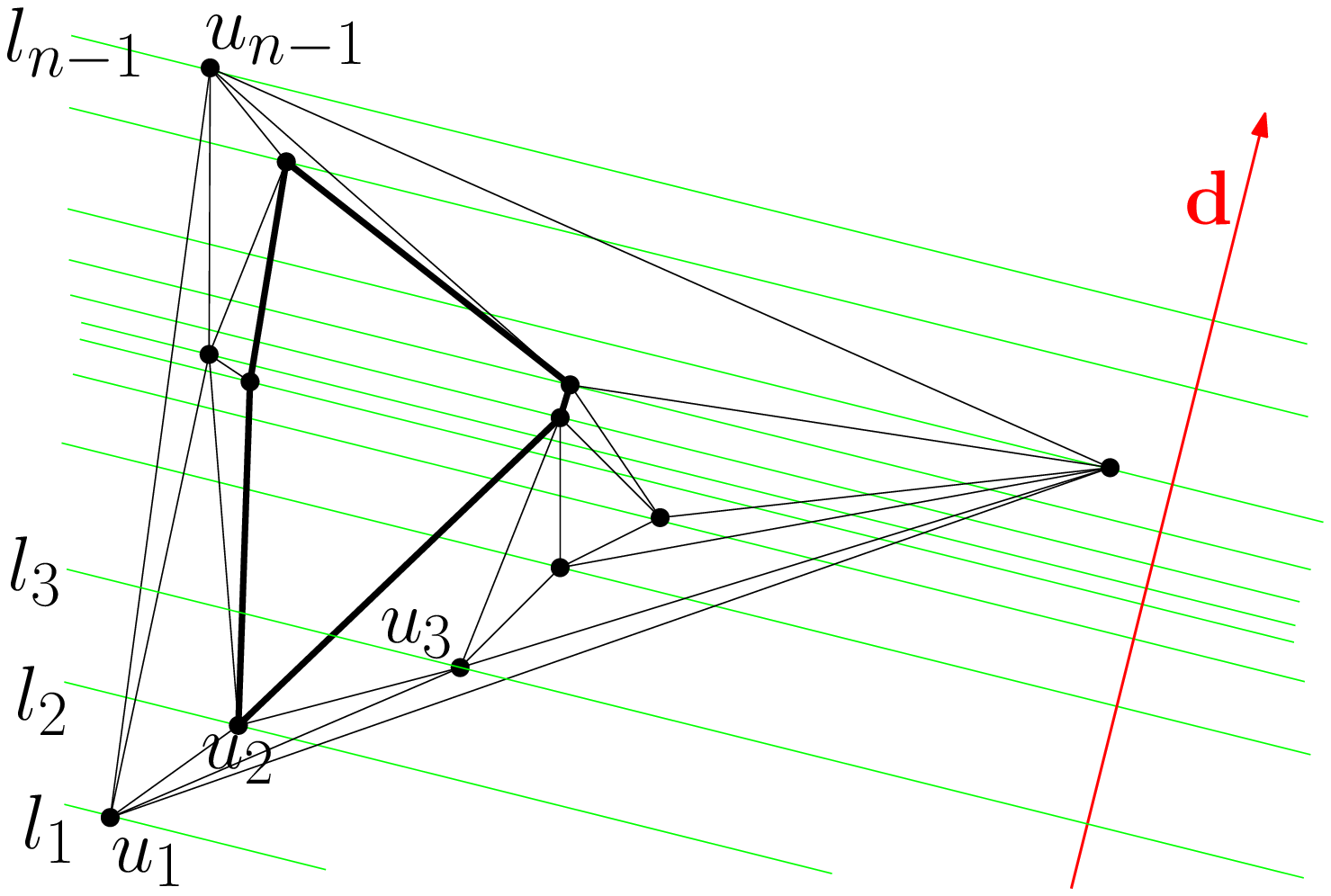}} \\
(a) \hspace{1mm} & (b) \hspace{1mm} & (c)
\end{tabular}
\caption{(a) Straight-line planar drawing $\Gamma$ of $G$. (b) Straight-line level planar drawing $\Gamma'$ of $(G',\vec d,L',\gamma')$. (c) Convex straight-line level planar drawing $\Gamma'_M$ of $(G',\vec d,L',\gamma')$.}
\label{fig:convexifier-process}
\end{center}
\end{figure}

Let $G'$ be the $(n-1)$-vertex plane graph obtained by removing $v$ and its incident edges from $G$. Also, let $\Gamma'$ be the straight-line planar drawing of $G'$ obtained by removing $v$ and its incident edges from $\Gamma$. We have the following lemma.

\begin{lemma} \label{le:triconnected-after-removal}
Graph $G'$ is $3$-connected.
\end{lemma}
\begin{proof}
Suppose, for a contradiction, that $G'$ contains a set $S'$ of vertices with $|S'|\leq 2$ whose removal disconnects $G'$. It follows that removing the vertices in $S=S'\cup\{v\}$ from $G$ disconnects $G$. If $|S|=1$ or $|S|=2$, then $G$ contains a separation $1$-set or $2$-set, respectively, in both cases contradicting the fact that $G$ is a maximal plane graph. If $|S|=3$, then $S$ is a separating $3$-set. However, any separating $3$-set in a maximal plane graph induces a separating $3$-cycle $C$. Hence, $C$ contains at least one neighbor of $v$ in its interior and at least one neighbor of $v$ in its exterior. This contradicts the assumption that $v$ is a quasi-contractible vertex of $G$.
\end{proof}

Consider the polygon $Q_v$ representing $C_v$ in $\Gamma$ and in $\Gamma'$. By Lemma~\ref{le:monotone-in-one-direction}, $Q_v$ is $\vec d$-monotone, for some oriented straight line $\vec d$. Slightly perturb the slope of $\vec d$ so that no line through two vertices of $G$ in $\Gamma$ is perpendicular to $\vec d$. If the perturbation is small enough, then $Q_v$ is still $\vec d$-monotone.
Denote by $u_1,\dots,u_{n-1}$ the vertices of $G'$ ordered according to their projection on $\vec d$. For $1\leq i\leq n-1$, denote by $l_i$ the line through $u_i$ orthogonal to $\vec d$. Let $L'=\{l_1,\dots,l_{n-1}\}$; note that the lines in $L'$ are parallel and distinct. Let $\gamma'$ be the function that maps $u_i$ to $l_i$, for $1\leq i\leq n-1$. See Fig.~\ref{fig:convexifier-process}(b).

\begin{lemma} \label{le:orientation}
$(G',\vec d,L',\gamma')$ is a hierarchical plane st-graph.
\end{lemma}
\begin{proof}
By construction, $\Gamma'$ is a straight-line level planar drawing of $(G',\vec d,L',\gamma')$, hence $(G',\vec d,L',\gamma')$ is a hierarchical plane graph. Further, every polygon delimiting a face of $G'$ in $\Gamma'$ is $\vec d$-monotone. This is true for $Q_v$ by construction and for every other polygon $Q_i$ delimiting a face of $G'$ in $\Gamma'$ by Lemma~\ref{le:convex-is-monotone}, given that $Q_i$ is a triangle and hence it is convex. Since every polygon delimiting a face of $G'$ in $\Gamma'$ is $\vec d$-monotone, every face of $G'$ is an st-face, hence $(G',\vec d,L',\gamma')$ is a hierarchical plane st-graph.
\end{proof}

By Lemmata~\ref{le:triconnected-after-removal} and~\ref{le:orientation}, $(G',\vec d,L',\gamma')$ is a triconnected hierarchical plane st-graph. By Theorem~\ref{th:hong-nagamochi}, a convex straight-line level planar drawing $\Gamma'_M$ of $(G',\vec d,L',\gamma')$ exists. Denote by $Q^M_v$ the convex polygon representing $C_v$ in $\Gamma'_M$. See Fig.~\ref{fig:convexifier-process}(c).

Denote by $r$ and $s$ the minimum and the maximum index such that $u_r$ and $u_s$ belong to $C_v$, respectively. Denote by $l(v)$ the line through $v$ orthogonal to $\vec d$ in $\Gamma$. If $l(v)$ were contained in the half-plane delimited by $l_r$ and not containing $l_s$, then $v$ would not lie inside $Q_v$ in $\Gamma$, as the projection of every vertex of $Q_v$ on $\vec d$ would follow the projection of $v$ on $\vec d$. Analogously, $l(v)$ is not contained in the half-plane delimited by $l_s$ and not containing $l_r$. It follows that $l(v)$ is ``in-between'' $l_r$ and $l_s$, that is, $l(v)$ lies in the strip defined by $l_r$ and $l_s$.

Construct a straight-line planar drawing $\Gamma_M$ of $G$ from $\Gamma'_M$ by placing $v$ on any point at the intersection of $l(v)$ and the interior of $Q^M_v$. Observe that such an intersection is always non-empty, given that $l_r$ and $l_s$ have non-empty intersection with $Q^M_v$, given that  $l(v)$ is in-between $l_r$ and $l_s$, and given that $Q^M_v$ is a convex polygon.

\begin{figure}[htb]
\begin{center}
\mbox{\includegraphics[scale=0.42]{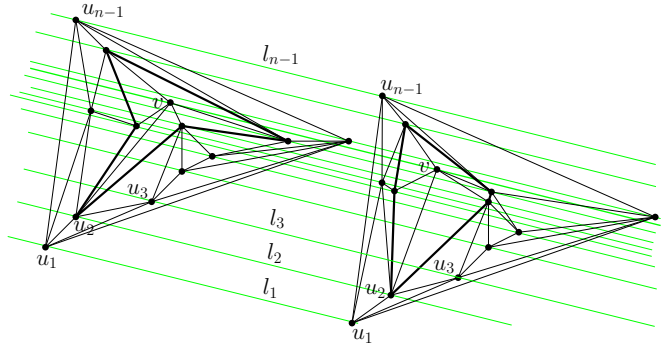}}
\caption{Morphing $\Gamma$ into a straight-line planar drawing $\Gamma_M$ of $G$ in which the polygon $Q^M_v$ representing $C_v$ is convex. The thick green line is $l(v)$.}
\label{fig:convexifier-morphing}
\end{center}
\end{figure}

Let $\gamma$ be the function that maps $v$ to $l(v)$ and $u_i$ to $l_i$, for $1\leq i\leq n-1$. We have that $\Gamma$ and $\Gamma_M$ are left-to-right equivalent straight-line level planar drawings of $(G,\vec d,L'\cup \{l(v)\},\gamma)$. By Lemma~\ref{le:unidirectional}, the linear morph transforming $\Gamma$ into $\Gamma_M$ is planar and unidirectional. Further, the polygon $Q^M_v$ representing $C_v$ in $\Gamma_M$ is convex. Thus, algorithm {\sc fast convexifier} consists of a single unidirectional morphing step transforming $\Gamma$ into $\Gamma_M$. See Fig.~\ref{fig:convexifier-morphing}.

\subsection{Contractibility Creator} \label{subse:contractibility}

In this section we describe an algorithm, called {\sc contractibility creator}, that receives a straight-line planar drawing $\Gamma$ of a plane graph $G$, a quasi-contractible vertex $v$ of $G$, and a neighbor $x$ of $v$, and returns a planar unidirectional morph with $O(1)$ morphing steps transforming $\Gamma$ into a straight-line planar drawing $\Gamma'$ of $G$ in which $v$ is $x$-contractible.

Denote by $u_1,\dots,u_k$ the clockwise order of the neighbors of $v$. If $k=1$, then $v$ is $x$-contractible in $\Gamma$, hence algorithm {\sc contractibility creator} returns $\Gamma'=\Gamma$.

If $k\geq 2$, consider any pair of consecutive neighbors of $v$, say $u_i$ and $u_{i+1}$ (where $u_{k+1}=u_1$). See Fig.~\ref{fig:contractibility}(a). If edge $(u_i,u_{i+1})$ belongs to $G$, then cycle $(u_i,v,u_{i+1})$ delimits a face of $G$, given that $v$ is quasi-contractible. Otherwise, we aim at morphing $\Gamma$ into a straight-line planar drawing of $G$ where a dummy edge $(u_i,u_{i+1})$ can be introduced while maintaining planarity and while ensuring that cycle $(u_i,v,u_{i+1})$ delimits a face of the augmented graph $G\cup\{(u_i,u_{i+1})\}$. This is accomplished as follows:

\begin{figure}[htb]
\begin{center}
\begin{tabular}{c}
\begin{tabular}{c c c}
\mbox{\includegraphics[scale=0.285]{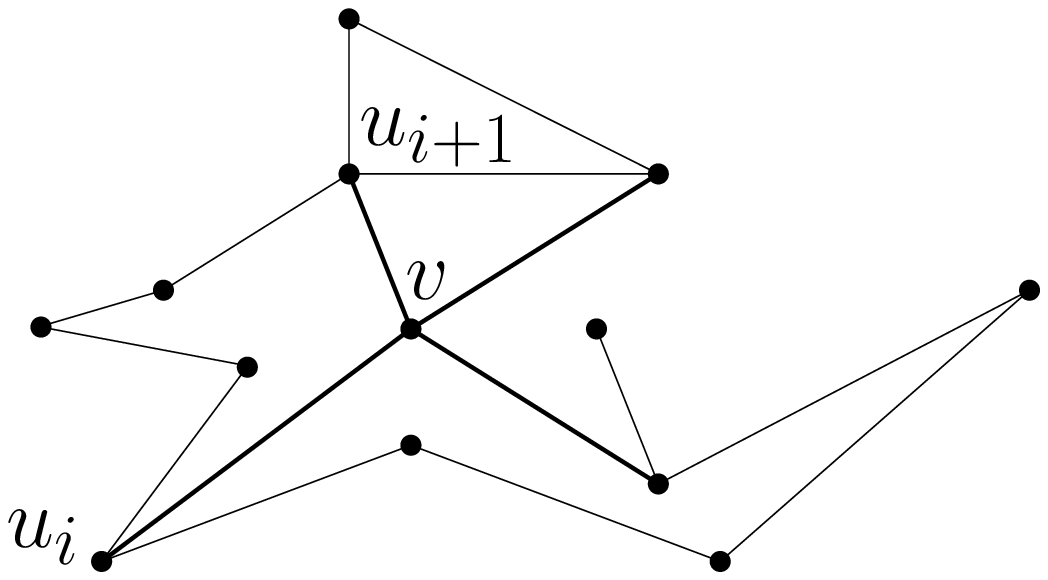}} &
\mbox{\includegraphics[scale=0.285]{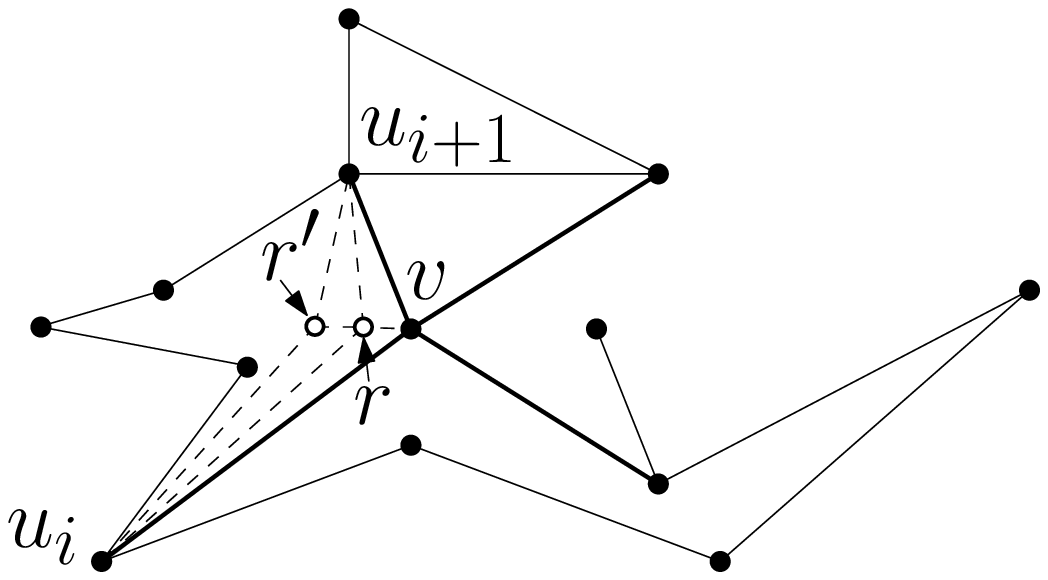}} &
\mbox{\includegraphics[scale=0.285]{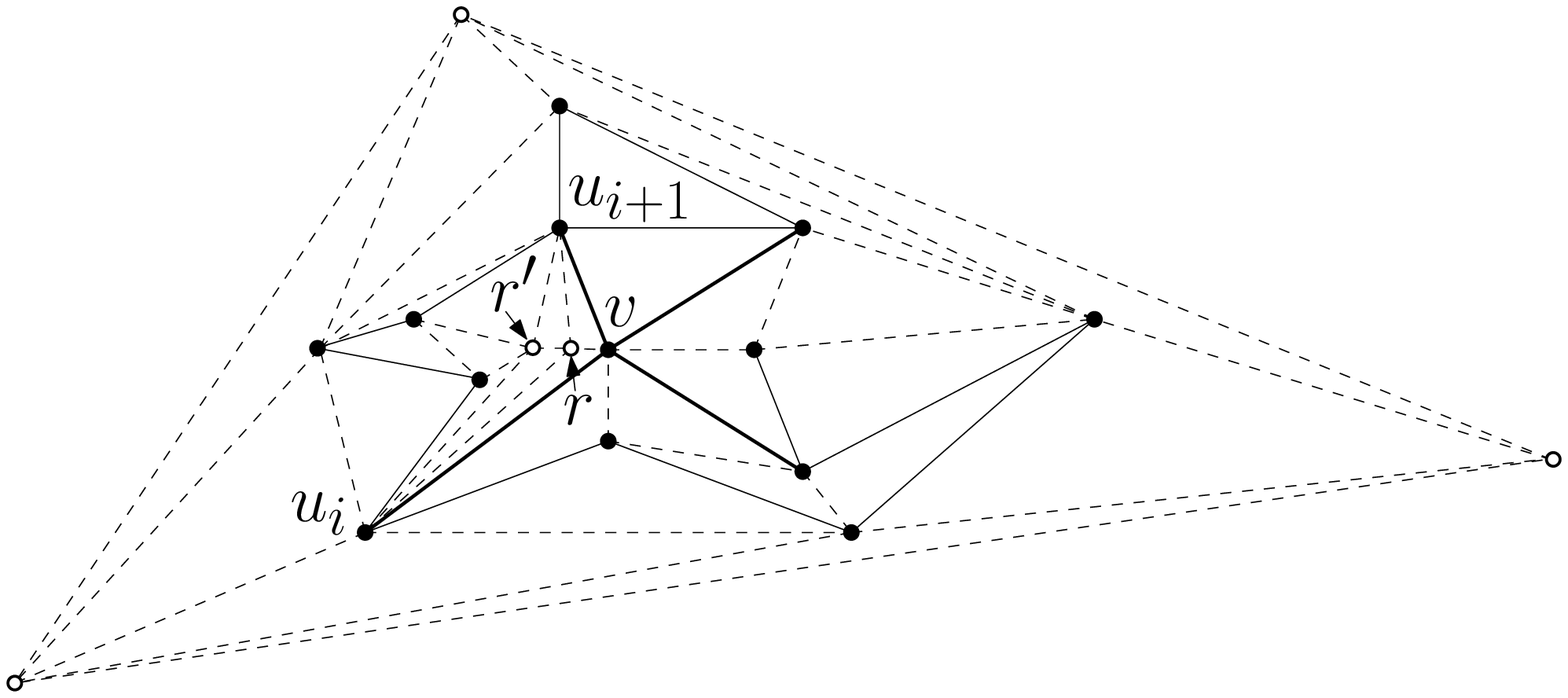}}\\
(a) & (b) & (c)
\end{tabular}\\
\begin{tabular}{c c}
\mbox{\includegraphics[scale=0.285]{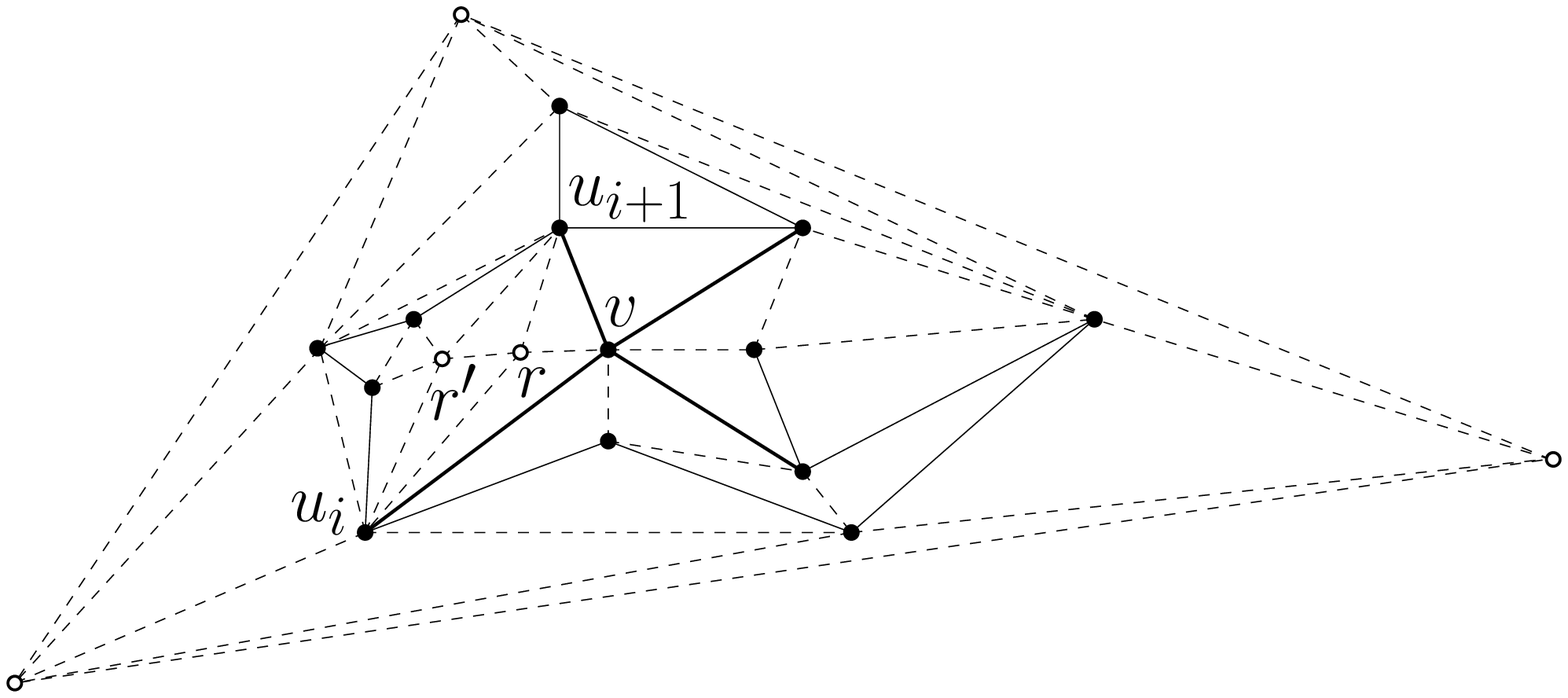}} \hspace{5mm} &
\mbox{\includegraphics[scale=0.285]{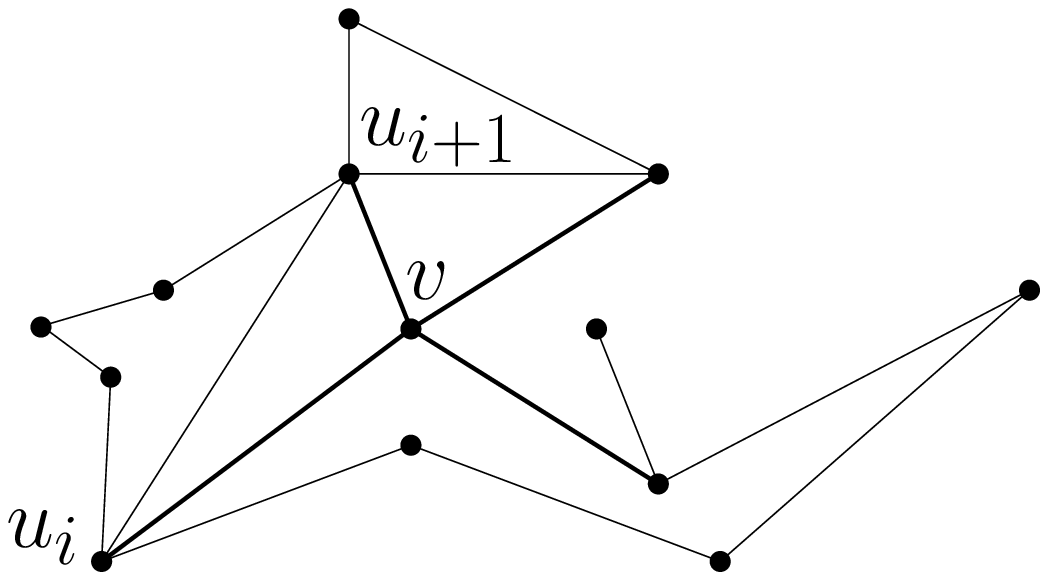}} \\
(d) \hspace{5mm} & (e)
\end{tabular}
\end{tabular}
\caption{(a) Drawing $\Gamma$ of $G$. (b) Drawing $\Gamma^+$ of $G^+$. (c) Drawing $\Gamma^*$ of $G^*$. (d) Drawing $\Gamma^*_M$ of $G^*$. (e) Drawing $\Gamma_M$ of $G\cup\{(u_i,u_{i+1})\}$.}
\label{fig:contractibility}
\end{center}
\end{figure}

\begin{enumerate}
\item We add two dummy vertices $r$ and $r'$, and six dummy edges $(r,v)$, $(r,u_i)$, $(r,u_{i+1})$, $(r',u_i)$, $(r',u_{i+1})$, and $(r,r')$ to $\Gamma$ and $G$, obtaining a straight-line planar drawing $\Gamma^+$ of a plane graph $G^+$, in such a way that $\Gamma^+$ is planar and cycles $(v,r,u_i)$, $(v,r,u_{i+1})$, $(r',r,u_i)$, and $(r',r,u_{i+1})$ delimit faces of $G^+$. See Fig.~\ref{fig:contractibility}(b).
\item We add dummy vertices and edges to $\Gamma^+$ and $G^+$, obtaining a straight-line planar drawing $\Gamma^*$ of a graph $G^*$, in such a way that $\Gamma^*$ is planar, that $G^*$ is a maximal planar graph, and that edges $(u_i,u_{i+1})$ and $(r',v)$ do not belong to $G^*$. Observe that $r$ is a quasi-contractible vertex of $G^*$. See Fig.~\ref{fig:contractibility}(c).
\item We apply algorithm {\sc fast convexifier} to morph $\Gamma^*$ with one unidirectional morphing step into a straight-line planar drawing $\Gamma^*_M$ of $G^*$ such that the polygon of the neighbors of $r$ is convex. See Fig.~\ref{fig:contractibility}(d).
\item We remove from $\Gamma^*_M$ all the dummy vertices and edges that belong to $G^*$ and do not belong to $G$, and we add edge $(u_i,u_{i+1})$ to $\Gamma^*_M$ and $G$, obtaining a straight-line planar drawing $\Gamma_M$ of graph $G\cup\{(u_i,u_{i+1})\}$. See Fig.~\ref{fig:contractibility}(e).
\end{enumerate}

If $k=2$, then after the above described algorithm is performed, we have that $v$ is $x$-contractible in $\Gamma'=\Gamma_M$, both if $x=u_1$ or if $x=u_2$, given that $(v,u_1,u_2)$ delimits a face of $G\cup\{(u_1,u_2)\}$. If $3\leq k\leq 5$, then the above described algorithm is repeated at most $k$ times (namely once for each pair of consecutive neighbors of $v$ that are not adjacent in $G$), at each time inserting an edge between a distinct pair of consecutive neighbors of $v$. Eventually, we obtain a straight-line planar drawing $\Phi$ of plane graph $G\cup\{(u_1,u_2),(u_2,u_3),(u_3,u_4),(u_4,u_5),(u_5,u_1)\}$ in which $v$ is quasi-contractible. Then we add dummy vertices and edges to $\Phi$, obtaining a straight-line planar drawing $\Sigma$ of a graph $H$, in such a way that $H$ is a maximal planar graph and that $v$ is quasi-contractible in $\Sigma$. We apply algorithm {\sc fast convexifier} to morph $\Sigma$ with one unidirectional morphing step into a straight-line planar drawing $\Psi$ of $H$ such that the polygon of the neighbors of $v$ is convex. Hence, $v$ is contractible onto any of its neighbors in $\Psi$. Then, we remove the edges of $H$ not in $G$, obtaining a straight-line planar drawing $\Gamma'$ of $G$ in which $v$ is contractible onto any of its neighbors; hence, $v$ is $x$-contractible in $\Gamma'$. Finally, observe that $\Gamma'$ is obtained from $\Gamma$ in at most $k+1\leq 6$ unidirectional morphing steps.

\subsection{The Algorithm} \label{subse:algorithm}

We now describe an algorithm to construct a pseudo-morph $\cal P$ with $O(n)$ steps between any two straight-line planar drawings $\Gamma_s$ and $\Gamma_t$ of the same $n$-vertex plane graph $G$.

The algorithm works by induction on $n$. If $n=1$, then ${\cal P}$ consists of a single unidirectional morphing step transforming $\Gamma_s$ into $\Gamma_t$. If $n\geq 2$, then let $v$ be a quasi-contractible vertex of $G$, which exists by Lemma~\ref{le:candidate_exists}, and let $x$ be any neighbor of $v$. Let $M_s$ and $M_t$ be the planar unidirectional morphs with $O(1)$ morphing steps produced by algorithm {\sc contractibility creator} transforming $\Gamma_s$ and $\Gamma_t$ into straight-line planar drawings $\Gamma^x_s$ and $\Gamma^x_t$ of $G$, respectively, such that $v$ is $x$-contractible both in $\Gamma^x_s$ and in $\Gamma^x_t$. Let $G'$ be the $(n-1)$-vertex plane graph obtained by contracting $v$ onto $x$ in $G$, and let $\Gamma_s'$ and $\Gamma_t'$ be the straight-line planar drawings of $G'$ obtained from $\Gamma^x_s$ and $\Gamma^x_t$, respectively, by contracting $v$ onto $x$. Further, let ${\cal P}'$ be the inductively constructed pseudo-morph between $\Gamma_s'$ and $\Gamma_t'$. Then, pseudo-morph $\cal P$ is defined as the unidirectional morph $M_s$ transforming $\Gamma_s$ into $\Gamma^x_s$, followed by the contraction of $v$ onto $x$ in $\Gamma^x_s$, followed by the pseudo-morph ${\cal P}'$ between $\Gamma_s'$ and $\Gamma_t'$, followed by the uncontraction of $v$ from $x$ into $\Gamma^x_t$, followed by the unidirectional morph $M^{-1}_t$ transforming $\Gamma^x_t$ into $\Gamma_t$. Observe that $\cal P$ has a number of steps which is a constant plus the number of steps of ${\cal P}'$. Hence, ${\cal P}$ consists of $O(n)$ steps.

A unidirectional planar morph $M$ between $\Gamma_s$ and $\Gamma_t$ can be constructed with a number of morphing steps equal to the number of steps of $\cal P$, by Lemma~\ref{le:pseudo-to-morph}. This proves the following:

\begin{theorem}\label{th:main}
Let $\Gamma_s$ and $\Gamma_t$ be any two straight-line planar drawings of the same $n$-vertex plane graph $G$. There exists an algorithm to construct a planar unidirectional morph with $O(n)$ morphing steps transforming $\Gamma_s$ into $\Gamma_t$.
\end{theorem}

\section{A Lower Bound} \label{se:lower}

In this section we show two straight-line planar drawings $\Gamma_s$ and $\Gamma_t$ of an $n$-vertex path $P=(v_1,\dots,v_n)$, and we prove that any planar morph $M$ between $\Gamma_s$ and $\Gamma_t$ requires $\Omega(n)$ morphing steps. In order to simplify the description, we consider each edge $e_i=(v_i,v_{i+1})$ as oriented from $v_i$ to $v_{i+1}$, for $i=1,\dots,n-1$.

Drawing $\Gamma_s$ (see Fig.~\ref{fig:lb-path}) is such that all the vertices of $P$ lie on a horizontal straight-line with $v_i$ to the left of $v_{i+1}$, for each $i = 1, \dots, n-1$.

Drawing $\Gamma_t$ (see Fig.~\ref{fig:lb-spiral}) is such that:
\begin{itemize}
\item for each $i = 1, \dots, n-1$ with $i \textrm{ mod } 3 \equiv 1$, the (green) segment representing $e_i$ is horizontal with $v_i$ to the left of $v_{i+1}$;
\item for each $i = 1, \dots, n-1$ with $i \textrm{ mod } 3 \equiv 2$, the (blue) segment representing $e_i$ is parallel to line $y=\tan(\frac{2\pi}{3})x$ with $v_i$ to the right of $v_{i+1}$; and
\item for each $i = 1, \dots, n-1$ with $i \textrm{ mod } 3 \equiv 0$, the (red) segment representing $e_i$ is parallel to line $y=\tan(-\frac{2\pi}{3})x$ with $v_i$ to the right of $v_{i+1}$.
\end{itemize}

\begin{figure}[h]
\centering
\subfigure[]{\includegraphics[scale=.5]{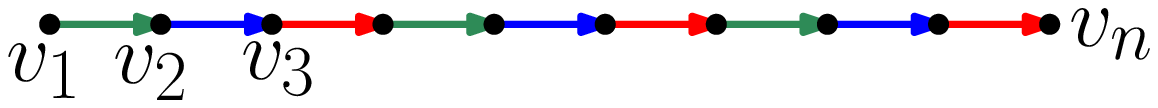}\label{fig:lb-path}} \hspace{1cm}
\subfigure[]{\includegraphics[scale=.5]{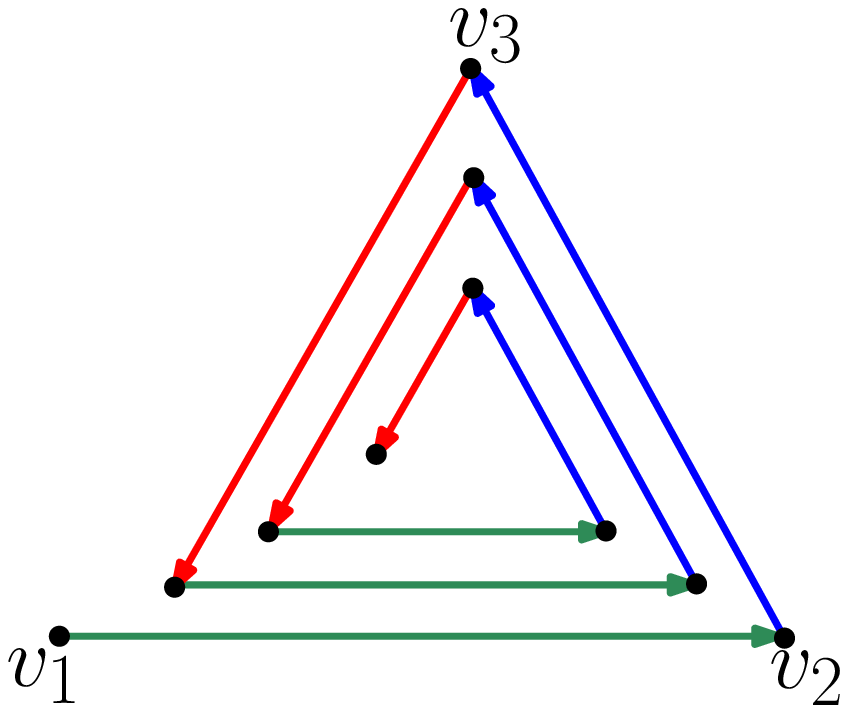}\label{fig:lb-spiral}}
\caption{Drawings $\Gamma_s$ (a) and $\Gamma_t$ (b).}
\label{fig:lb-drawings}
\end{figure}

Let $M=\morph{\Gamma_s=\Gamma_1,\dots,\Gamma_m=\Gamma_t}$ be any planar morph transforming $\Gamma_s$ into $\Gamma_t$.

For $i=1,\dots,n$ and $j=1,\dots,m$, we denote by $v_i^j$ the point where vertex $v_i$ is placed in $\Gamma_j$; also, for $i=1,\dots,n-1$ and $j=1,\dots,m$ we denote by $e_i^j$ the directed straight-line segment representing edge $e_i$ in $\Gamma_j$.

For $1\leq j\leq m-1$, we define the \emph{rotation} $\rho_i^j$ of $e_i$ around $v_i$ during the morphing step $\langle \Gamma_{j},\Gamma_{j+1} \rangle$ as follows (see Fig.~\ref{fig:lb-rotation}). Translate $e_i$ at any time instant of $\langle \Gamma_{j},\Gamma_{j+1} \rangle$ so that $v_i$ stays fixed at a point $a$ during the entire morphing step. After this translation, the morph between $e_i^{j}$ and $e_i^{j+1}$ is a rotation of $e_i$ around $a$ (where $e_i$ might vary its length during $\langle \Gamma_{j},\Gamma_{j+1} \rangle$) spanning an angle $\rho_i^j$, where we assume $\rho_i^j>0$ if the rotation is counter-clockwise, and $\rho_i^j<0$ if the rotation is clockwise. We have the following.


\begin{figure}[h!]
\centering
\hfill
\subfigure[]{\includegraphics[scale=.4]{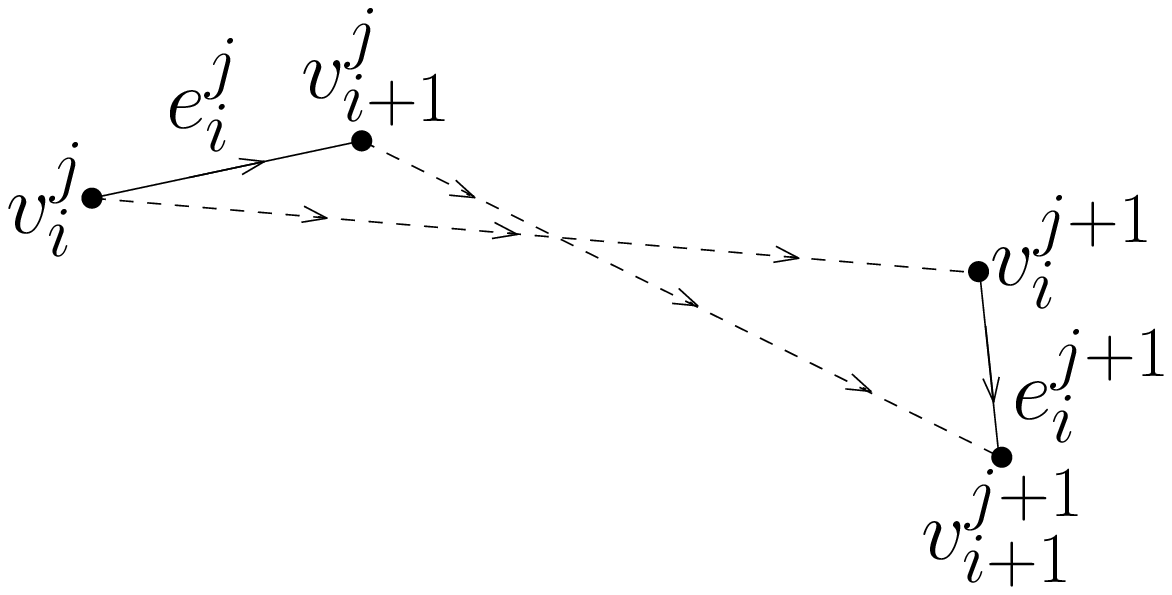}\label{fig:lb-rotation-a}}
\hfill
\subfigure[]{\includegraphics[scale=.4]{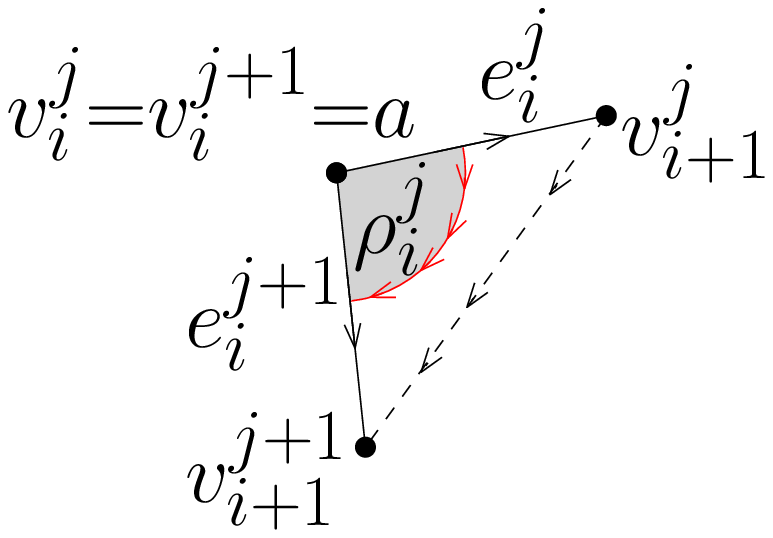}\label{fig:lb-rotation-c}}
\hfill
\caption{Rotation $\rho_i^j$. (a) Morph between $e_i^{j}$ and $e_i^{j+1}$. (b) Translation of the positions of $e_i$ during $\langle \Gamma_{j},\Gamma_{j+1} \rangle$, resulting in $e_i$ spanning an angle $\rho_i^j$ around $v_i$.}
\label{fig:lb-rotation}
\end{figure}

\begin{lemma}\label{le:lb-rotation}
For each $j=1,\dots,m-1$ and $i=1,\dots,n-1$, we have $|\rho_i^j| < \pi$.
\end{lemma}
\begin{proof}
Assume, for a contradiction, that $|\rho_i^j| \geq \pi$, for some $1\leq j\leq x-1$ and $1\leq i\leq n-1$. Also assume, w.l.o.g., that the morphing step $\langle \Gamma_{j},\Gamma_{j+1} \rangle$ happens between time instants $t=0$ and $t=1$. For any $0\leq t\leq 1$, denote by $v_i(t)$, $v_{i+1}(t)$, $e_i(t)$, and $\rho_i^j(t)$ the position of $v_i$, the position of $v_{i+1}$, the drawing of $e_i$, and the rotation of $e_i$ around $v_i$ at time instant $t$, respectively. Note that $v_i(0)=v_i^j$, $v_{i+1}(0)=v_{i+1}^j$, $e_i(0)=e_i^j$, $\rho_i^j(0)=0$, and $\rho_i^j(1)=\rho_i^j$. Since a morph is a continuous transformation and since $|\rho_i^j| \geq \pi$, there exists a time instant $t_{\pi}$ with $0< t_{\pi}\leq 1$ such that $|\rho_i^j(t_{\pi})|=\pi$.

We prove that there exists a time instant $t_r$ with $0< t_r\leq t_{\pi}$ in which $v_i(t)$ and $v_{i+1}(t)$ coincide, thus contradicting the assumption that morph $\langle \Gamma_{j},\Gamma_{j+1} \rangle$ is planar.

\begin{figure}[h!]
\centering{
\includegraphics[scale=0.4]{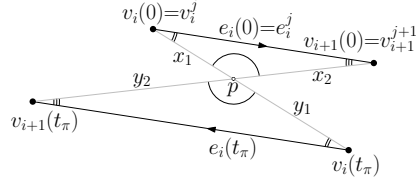}\label{fig:lb-rotation-proof-a}}
\caption{Illustration for the proof of Lemma~\ref{le:lb-rotation}.}
\label{fig:pi-rotation}
\end{figure}

Since $|\rho_i^j(t_{\pi})|=\pi$, it follows that $e_i(t_{\pi})$ is parallel to $e_i(0)$ and oriented in the opposite way. This easily leads to conclude that $t_r$ exists if $e_i(t_{\pi})$ and $e_i(0)$ are aligned. Otherwise, the straight-line segments $\overline{v_i(0) v_i(t_{\pi})}$ and $\overline{v_{i+1}(0) v_{i+1}(t_{\pi})}$ meet in a point $p$. Refer to Fig.~\ref{fig:pi-rotation}. Let $x_1=|\overline{p v_i(0)}|$, $x_2=|\overline{p v_{i+1}(0)}|$, $y_1=|\overline{p v_i(t_{\pi})}|$, and $y_2=|\overline{p v_{i+1}(t_{\pi})}|$. By the similarity of triangles $(v_i(0),p,v_{i+1}(0))$ and $(v_i(t_{\pi}),p,v_{i+1}(t_{\pi}))$, we have $\frac{x_1}{y_1}=\frac{x_2}{y_2}$ and hence $\frac{x_1}{x_1+y_1}=\frac{x_2}{x_2+y_2}$. Thus, $v_i(\frac{x_1}{x_1+y_1}t_{\pi})$ and $v_{i+1}(\frac{x_1}{x_1+y_1}t_{\pi})$ are coincident with $p$. This contradiction proves the lemma.
\end{proof}

For $j=1,\dots,m-1$, we denote by $M_j$ the subsequence $\morph{\Gamma_1,\dots,\Gamma_{j+1}}$ of $M$; also, for $i=1,\dots,n-1$, we define the \emph{total rotation} $\rho_i(M_j)$ of edge $e_i$ around $v_i$ during morph $M_j$ as $\rho_i(M_j)=\sum_{m=1}^{j}\rho_i^m$.

We will show in Lemma~\ref{le:lb-linear-total-rotation} that there exists an edge $e_i$, for some $1 \le i \le n-1$, whose total rotation $\rho_i(M_{m-1})=\rho_i(M)$ is $\Omega(n)$. In order to do that, we first analyze the relationship between the total rotation of two consecutive edges of $P$.

\begin{lemma}\label{le:lb-diff-rot}
For each $j=1,\dots,m-1$ and for each $i=1,\dots,n-2$, we have that $|\rho_{i+1}(M_j)-\rho_{i}(M_j)|<\pi$.
\end{lemma}
\begin{proof}
Suppose, for a contradiction, that $|\rho_{i+1}(M_j)-\rho_{i}(M_j)|\geq \pi$ for some $1\leq j \leq m-1$ and $1\leq i \leq n-2$. Assume that $j$ is minimal under this hypothesis. Since each vertex moves continuously during $M_j$, there exists an intermediate drawing $\Gamma^*$ of $P$, occurring during morphing step $\morph{\Gamma_j,\Gamma_{j+1}}$, such that $|\rho_{i+1}(M^*)-\rho_{i}(M^*)| = \pi$, where $M^*=\morph{\Gamma_1,\dots,\Gamma_j,\Gamma^*}$ is the morph obtained by concatenating $M_{j-1}$ with the morphing step transforming $\Gamma_{j}$ into $\Gamma^*$.
Recall that in $\Gamma_1$ edges $e_i$ and $e_{i+1}$ lie on the same straight line and have the same orientation. Then, since $|\rho_{i+1}(M^*)-\rho_{i}(M^*)| = \pi$, in $\Gamma^*$ edges $e_i$ and $e_{i+1}$ are parallel and have opposite orientations. Also, since edges $e_i$ and $e_{i+1}$ share vertex $v_{i+1}$, they lie on the same line. This implies that such edges overlap, contradicting the hypothesis that $M^*$, $M_j$, and $M$ are planar.
\end{proof}

We are now ready to prove the key lemma for the lower bound.

\begin{lemma}\label{le:lb-linear-total-rotation}
There exists an index $i$ such that $|\rho_i(M)| \in \Omega(n)$.
\end{lemma}
\begin{proof}
Refer to Fig.~\ref{fig:lb-drawings}. For every $1\leq i\leq n-2$, edges $e_i$ and $e_{i+1}$ form an angle of $\pi$ radiants in $\Gamma_s$, while they form an angle of $\frac{\pi}{3}$ radiants in $\Gamma_t$. Hence, $\rho_{i+1}(M)=\rho_i(M)+\frac{2\pi}{3} + 2z_i\pi$, for some $z_i\in \mathbb{Z}$.


In order to prove the lemma, it suffices to prove that $z_i=0$, for every $i=1,\dots,n-2$. Namely, in this case $\rho_{i+1}(M)=\rho_i(M)+\frac{2\pi}{3}$ for every $1\leq i\leq n-2$, and hence $\rho_{n-1}(M) = \rho_1(M) + \frac{2\pi}{3}(n-2)$. This implies $|\rho_{n-1}(M)-\rho_1(M)|\in \Omega(n)$, and thus $|\rho_1(M)|\in \Omega(n)$ or $|\rho_{n-1}(M)|\in \Omega(n)$.

Assume, for a contradiction, that $z_i\neq 0$, for some $1 \leq i \leq n-2$. If $z_i>0$, then $\rho_{i+1}(M)\geq \rho_i(M)+\frac{8\pi}{3}$; further, if $z_i<0$, then $\rho_{i+1}(M)\leq \rho_i(M)-\frac{4\pi}{3}$. Since each of these inequalities contradicts Lemma~\ref{le:lb-diff-rot}, the lemma follows.
\end{proof}

We are now ready to state the main theorem of this section.

\begin{theorem}\label{th:lb-bound}
There exists two straight-line planar drawings $\Gamma_s$ and $\Gamma_t$ of an $n$-vertex path $P$ such that any planar morph between $\Gamma_s$ and $\Gamma_t$ requires $\Omega(n)$ morphing steps.
\end{theorem}
\begin{proof}
The two drawings $\Gamma_s$ and $\Gamma_t$ of path $P=(v_1,\dots,v_n)$ are those illustrated in Fig.~\ref{fig:lb-drawings}. By Lemma~\ref{le:lb-linear-total-rotation}, there exists an edge $e_i$ of $P$, for some $1 \le i \le n-1$, such that $|\sum_{j=1}^{x-1}\rho_i^j| \in \Omega(n)$. Since, by Lemma~\ref{le:lb-rotation}, we have that $|\rho_i^j| < \pi$ for each $j=1,\dots,x-1$, it follows that $x \in \Omega(n)$. This concludes the proof of the theorem.
\end{proof}

\section{Conclusions} \label{se:conclusions}

In this paper we presented an algorithm to construct a planar morph between two planar straight-line drawings of the same $n$-vertex plane graph in $O(n)$ morphing steps. We also proved that this bound is tight (note that our lower bound holds for any morphing algorithm in which the vertex trajectories are polynomial functions of constant degree).

In our opinion, the main challenge in this research area is the one of designing algorithms to construct planar morphs between straight-line planar drawings with good resolution and within polynomial area (or to prove that no such algorithm exists). In fact, the algorithm we presented, as well as other algorithms known at the state of the art~\cite{aac-mpgdpns-13,afpr-mpgde-13,c-dprc-44,t-dpg-83}, construct intermediate drawings in which the ratio between the lengths of the longest and of the shortest edge is exponential. Guaranteeing good resolution and small area seems to be vital for making a morphing algorithm of practical utility.

Finally, we would like to mention an original problem that generalizes the one we solved in this paper and that we repute very interesting. Let $\Gamma_s$ and $\Gamma_t$ be two straight-line drawings of the same (possibly non-planar) topological graph $G$. Does a morphing algorithm exist that morphs $\Gamma_s$ into $\Gamma_t$ and that preserves the topology of the drawing at any time instant? A solution to this problem is not known even if we allow the trajectories followed by the vertices to be of arbitrary complexity.

\bibliography{bibliography}

\begin{thebibliography}{10}

\bibitem{aac-mpgdpns-13}
S.~Alamdari, P.~Angelini, T.~M. Chan, G.~{Di Battista}, F.~Frati, A.~Lubiw,
  M.~Patrignani, V.~Roselli, S.~Singla, and B.~T. Wilkinson.
\newblock Morphing planar graph drawings with a polynomial number of steps.
\newblock In S.~Khanna, editor, {\em 24th Annual ACM-SIAM Symposium on Discrete
  Algorithms (SODA '13)}, pages 1656--1667. SIAM, 2013.

\bibitem{afpr-mpgde-13}
P.~Angelini, F.~Frati, M.~Patrignani, and V.~Roselli.
\newblock Morphing planar graph drawings efficiently.
\newblock In S.~Wismath and A.~Wolff, editors, {\em 21st International
  Symposium on Graph Drawing (GD '13)}, volume 8242 of {\em LNCS}, pages
  49--60. Springer, 2013.

\bibitem{bhl-mpgdum-13}
F.~{Barrera-Cruz}, P.~Haxell, and A.~Lubiw.
\newblock Morphing planar graph drawings with unidirectional moves.
\newblock Mexican Conference on Discr. Math. and Comput. Geom., 2013.

\bibitem{c-dprc-44}
S.~S. Cairns.
\newblock Deformations of plane rectilinear complexes.
\newblock {\em American Math. Monthly}, 51:247--252, 1944.

\bibitem{cyn-lacdp-84}
N.~Chiba, T.~Yamanouchi, and T.~Nishizeki.
\newblock Linear algorithms for convex drawings of planar graphs.
\newblock In J.~A. Bondy and U.~S.~R. Murty, editors, {\em Progress in Graph
  Theory}, pages 153--173. Academic Press, New York, NY, 1984.

\bibitem{ekp-ifmpg-03}
C.~Erten, S.~G. Kobourov, and C.~Pitta.
\newblock Intersection-free morphing of planar graphs.
\newblock In {\em 11th Symposium on Graph Drawing}, pages 320--331, 2003.

\bibitem{fe-gdm-02}
C.~Friedrich and P.~Eades.
\newblock Graph drawing in motion.
\newblock {\em J. Graph Algorithms Appl.}, 6(3):353--370, 2002.

\bibitem{gs-gifpm-01}
C.~Gotsman and V.~Surazhsky.
\newblock Guaranteed intersection-free polygon morphing.
\newblock {\em Computers {\&} Graphics}, 25(1):67--75, 2001.

\bibitem{gs-tgopg-81}
B.~Grunbaum and G.C. Shephard.
\newblock {\em The geometry of planar graphs}.
\newblock Cambridge University Press, 1981.

\bibitem{hn-cdhpgcpg-10}
S.~H. Hong and H.~Nagamochi.
\newblock Convex drawings of hierarchical planar graphs and clustered planar
  graphs.
\newblock {\em J. Discrete Algorithms}, 8(3):282--295, 2010.

\bibitem{sg-cmcpt-01}
V.~Surazhsky and C.~Gotsman.
\newblock Controllable morphing of compatible planar triangulations.
\newblock {\em ACM Trans. Graph}, 20(4):203--231, 2001.

\bibitem{sg-imct-03}
V.~Surazhsky and C.~Gotsman.
\newblock Intrinsic morphing of compatible triangulations.
\newblock {\em Internat. J. of Shape Model.}, 9:191--201, 2003.

\bibitem{t-dpg-83}
C.~Thomassen.
\newblock Deformations of plane graphs.
\newblock {\em Journal of Combinatorial Theory, Series B}, 34(3):244--257,
  1983.

\bibitem{t-prg-84}
C.~Thomassen.
\newblock Plane representations of graphs.
\newblock In J.~A. Bondy and U.~S.~R. Murty, editors, {\em Progress in Graph
  Theory}, pages 43--69. Academic Press, New York, NY, 1984.

\end{thebibliography}
\bibliographystyle{plain}

\end{document}